\documentclass[a4paper,UKenglish,final]{llncs}

\usepackage{amsfonts,amssymb,amsmath}
\usepackage{stmaryrd}
\usepackage{xspace}
\usepackage{xparse}
\usepackage[retainorgcmds]{IEEEtrantools}
\usepackage{relsize}
\usepackage[basic]{complexity}
\renewclass{\EXP}{EXPTIME}
\usepackage{tikz}
\usepackage{colortbl}
\usepackage{multirow}
\usetikzlibrary{positioning,automata,arrows,decorations.pathreplacing,fit}

\tikzstyle{every state}=[minimum size=2em]

\tikzset{>=latex,auto,node distance=3cm, every
  loop/.style={looseness=6}, initial text={}, inner sep=1mm,
  loopright/.style={loop,looseness=6,out=35, in=-35},
  loopleft/.style={loop,looseness=6,out=145, in=215},
  loopabove/.style={loop,looseness=6,out=125, in=55},
  loopbelow/.style={loop,looseness=6,out=-125, in=-55}, }

\usepackage{paralist}
\usepackage[ruled,boxed,vlined,linesnumbered]{algorithm2e}
\usepackage[strict]{changepage}

\newcommand\parag[1]{\smallskip\noindent\textbf{#1}}

\title{Real-time Synthesis is Hard!}

\author{Thomas Brihaye\inst{1}, Morgane Esti\'evenart\inst{1}, Gilles
  Geeraerts\inst{2}, Hsi-Ming Ho\inst{1}, Benjamin Monmege\inst{3}, Nathalie Sznajder\inst{4}}

\institute{Universit\'e de Mons, Belgium,
  \email{thomas.brihaye,morgane.estievenart,hsi-ming.ho@umons.ac.be}
  \and
  Universit\'e libre de Bruxelles, Belgium, 
  \email{gigeerae@ulb.ac.be}
  \and
  Aix Marseille Univ, CNRS, LIF, France,
  \email{benjamin.monmege@lif.univ-mrs.fr}
  \and 
  Sorbonne Universit\'es, UPMC, LIP6, France,
  \email{nathalie.sznajder@lip6.fr}
}

\newcommand\LTL{\ensuremath{\mathsf{LTL}}\xspace}
\newcommand\MTL{\ensuremath{\mathsf{MTL}}\xspace}
\newcommand\ECL{\ensuremath{\mathsf{ECL}}\xspace}
\newcommand\ECLfut{\ensuremath{\ECL_\textsf{fut}}\xspace}
\newcommand\MITL{\ensuremath{\mathsf{MITL}}\xspace}
\newcommand\SMTL{\ensuremath{\mathsf{Safety\text{-}MTL}}\xspace}
\newcommand\CFMTL{\ensuremath{\mathsf{coFlat\text{-}MTL}}\xspace}
\ProvideDocumentCommand{\Until}{o}{\IfNoValueTF{#1}{\mathop{\mathsf{U}}}{\mathop{\mathsf{U}_{#1}}}}
\ProvideDocumentCommand{\dualUntil}{o}{\IfNoValueTF{#1}{\mathop{\widetilde{\mathsf{U}}}}{\mathop{\widetilde{\mathsf{U}}_{#1}}}}
\ProvideDocumentCommand{\Finally}{o}{\IfNoValueTF{#1}{\mathop{\lozenge}}{\mathop{\lozenge_{#1}}}}
\ProvideDocumentCommand{\Globally}{o}{\IfNoValueTF{#1}{\mathop{\square}}{\mathop{\square_{#1}}}}
\ProvideDocumentCommand{\Next}{o}{\IfNoValueTF{#1}{\mathop{\bigcirc}}{\mathop{\bigcirc_{#1}}}}

\ProvideDocumentCommand{\wUntil}{o}{\IfNoValueTF{#1}{\mathop{\overline{\mathsf{U}}}}{\mathop{\overline{\mathsf{U}}_{#1}}}}
\ProvideDocumentCommand{\wFinally}{o}{\IfNoValueTF{#1}{\mathop{\overline{\lozenge}}}{\mathop{\overline{\lozenge}_{#1}}}}
\ProvideDocumentCommand{\wGlobally}{o}{\IfNoValueTF{#1}{\mathop{\overline{\square}}}{\mathop{\overline{\square}_{#1}}}}

\newcommand\TA{\ensuremath{\mathrm{TA}}\xspace}
\newcommand\OCATA{\textsf{OCATA}\xspace}

\newcommand\Lang{\mathcal L}

\newcommand\T{\mathcal T} 
\newcommand\A{\mathcal A} 
\newcommand\B{\mathcal B} 
\newcommand\Plant{\mathcal P} 
\newcommand\transitions{\Delta}
\newcommand\dtransitions{\delta} 
\newcommand\atransitions{\delta} 
\newcommand\Operations{\Gamma}
\newcommand\operation{\gamma}
\newcommand\ptransitions{\delta_\Plant} 
\newcommand\Enabled{\mathsf{En}} 
\newcommand\Regions{REG}
\newcommand\reg{reg}
\newcommand\fract{\mathsf{fract}}
\newcommand\wordreg{H}
\newcommand\setwordreg{\mathcal C}
\newcommand\powerset[1]{2^{#1}}
\newcommand\Guards{\mathcal G}
\newcommand\AtomicGuards{\mathcal G^{\mathrm{atom}}}

\newcommand\strategy{\pi}

\newcommand\sem[1]{\ensuremath{\left\llbracket#1\right\rrbracket}}

\ProvideDocumentCommand{\ReacSynt}{s}{\IfBooleanT{#1}{\ensuremath{\mathsf{RSP}}}\IfBooleanF{#1}{\ensuremath{\mathsf{RS}}}\xspace}
\ProvideDocumentCommand{\ImplReacSynt}{s}{\IfBooleanT{#1}{\ensuremath{\mathsf{IRSP}}}\IfBooleanF{#1}{\ensuremath{\mathsf{IRS}}}\xspace}
\ProvideDocumentCommand{\BoundReacSynt}{s}{\IfBooleanT{#1}{\ensuremath{\mathsf{BRSP}}}\IfBooleanF{#1}{\ensuremath{\mathsf{BRS}}}\xspace}
\ProvideDocumentCommand{\BoundPrecReacSynt}{s}{\IfBooleanT{#1}{\ensuremath{\mathsf{BPRSP}}}\IfBooleanF{#1}{\ensuremath{\mathsf{BPRS}}}\xspace}

\newcommand{\RS}{\ensuremath{\mathsf{RS}}\xspace}
\newcommand{\IRS}{\ensuremath{\mathsf{IRS}}\xspace}
\newcommand{\BPRS}{\ensuremath{\mathsf{BPrecRS}}\xspace}
\newcommand{\BCRS}{\ensuremath{\mathsf{BClockRS}}\xspace}
\newcommand{\BRRS}{\ensuremath{\mathsf{BResRS}}\xspace}

\renewcommand\R{\mathbb R}
\newcommand\Q{\mathbb Q}
\newcommand\N{\mathbb N}
\renewcommand\phi{\varphi}
\renewcommand\leq{\leqslant}
\renewcommand\geq{\geqslant}

\newcommand\tuple[1]{\left\langle #1 \right\rangle}

\newcommand\checkl{\ensuremath{\textit{Check}^\leftarrow}}
\newcommand\checkr{\ensuremath{\textit{Check}^\rightarrow}}
\newcommand\psifst{\ensuremath{\psi^\rightarrow_\textit{fst}}}

\usepackage[pdftex,bookmarksopen=true,bookmarksopenlevel=2]{hyperref}

\begin{document}

\maketitle

\begin{abstract}
  We study the reactive synthesis problem (\RS) for specifications
  given in Metric Interval Temporal Logic (\MITL). \RS is known to be
  undecidable in a very general setting, but on infinite words only;
  and only the very restrictive \BRRS subcase is known to be decidable
  (see D'Souza \textit{et al.} and Bouyer \textit{et al.}). In this
  paper, we precise the decidability border of \MITL synthesis. We
  show \RS is undecidable on finite words too, and present a landscape
  of restrictions (both on the logic and on the possible controllers)
  that are still undecidable. On the positive side, we revisit \BRRS
  and introduce an efficient on-the-fly algorithm to solve it.
\end{abstract}

\section{Introduction} 
The design of programs that respect real-time specifications is a
difficult problem with recent and promising advances. Such programs
must handle thin timing behaviours, are prone to errors, and difficult
to correct a posteriori. Therefore, one road to the design of correct
real-time software is the use of automatic synthesis methods, that
\emph{build}, from a specification, a program which is correct by
construction. To this end, \emph{timed games} are nowadays recognised
as the key foundational model for the synthesis of real-time
programs. These games are played between a \emph{controller} and an
\emph{environment}, that propose actions in the system, modelled as a
\emph{plant}. The \emph{reactive synthesis problem} (\RS) consists,
given a real-time specification, in deciding whether the controller
has a winning strategy ensuring that every execution of the plant
consistent with this strategy (i.e., no matter the choices of the
environment) satisfies the specification. As an example, consider a
lift for which we want to design a software verifying certain safety
conditions. In this case, the plant is a (timed) automaton, whose
states record the current status of the lift (its floor, if it is
moving, the button on which users have pushed\ldots), as well as
timing information regarding the evolution in-between the different
states. On the other hand, the specification is usually given using
some real-time logic: in this work, we consider mainly specifications
given by a formula of \MITL \cite{AluFed96}, a real-time extension of
\LTL. Some actions in the plant are controllable (closing the doors,
moving the cart), while others belong to the environment (buttons
pushed by users, exact timing of various actions inside intervals,
failures\ldots). Then, the \RS problem asks to compute a controller
that performs controllable actions at the right moments, so that, for
all behaviours of the environment, the lift runs correctly.

In the \emph{untimed case}, many positive theoretical and practical
results have been achieved regarding \RS: for instance, when the
specification is given as an \LTL formula, we know that if a winning
strategy exists, then there is one that can be described by a finite
state machine \cite{PnuRos89}; and efficient \LTL synthesis algorithms
have been implemented \cite{FilJin09,BohyBFJR12}. Unfortunately, in
the real-time setting, the picture is not so clear. Indeed, a winning
strategy in a timed game might need unbounded memory to recall the
full prefix of the game, which makes the real-time synthesis problem a
\textit{hard} one. This is witnessed by three papers presenting
negative results: D'Souza and Madhusudan \cite{DSoMad02} and Bouyer
\textit{et al.} \cite{BouBoz06} show that \RS is undecidable (on
finite and infinite words) when the specification is respectively a
timed automaton and an \MTL formula (the two most expressive formalisms in
\figurename~\ref{fig:decsummary}). More recently, Doyen \textit{et
  al.} show \cite{DoyGee09} that \RS is undecidable in the infinite
words semantics, when the specification is given using \MITL; but
leave the finite words case open.

When facing an undecidability result, one natural research direction
consists in considering subcases in order to recover decidability:
here, this amounts to considering fragments of the logic, or
restrictions on the possible controllers. Such results can also be
found in the aforementioned works. In \cite{DSoMad02}, the authors
consider a variant of \RS, called \emph{bounded resources reactive
  synthesis} (\BRRS) where the number of clocks and the set of guards
that the controller can use are fixed a priori, and the specification
is given by means of a timed automaton. By coupling this technique
with the translation of \MITL into timed automata~\cite{BriEst13}, one
obtains a 3-$\EXP$ procedure (in the finite and infinite words
cases). Unfortunately, due to the high cost of translating \MITL into
timed automata and the need to construct its entire deterministic
region automaton, this algorithm is unlikely to be amenable to
implementation. Then, \cite{BouBoz06} presents an on-the-fly algorithm
for \BRRS with \MTL specifications (\MTL is a strict superset of
\MITL), on finite words, but their procedure runs in non-primitive
recursive time.

Hence, the decidability status of the synthesis problem (with \MITL
requirements) still raises several questions, namely:
\begin{inparaenum}[($i$)]
\item Can we relax the restrictions in the definition of \BRRS while
  retaining decidability?
\item Is \RS decidable on finite words, as raised in \cite{DoyGee09}?
\item Are there meaningful restrictions of the logic that make \RS
  decidable?
\item Can we devise an on-the-fly, efficient, algorithm that solves
  \BRRS in 3-$\EXP$ as in \cite{DSoMad02}?
\end{inparaenum}
In the present paper, we provide answers to those
questions.
First, we consider the additional \IRS, \BPRS and \BCRS problems, that
introduce different levels of restrictions. \IRS requests the
controller to be a timed automaton. \BPRS and \BCRS are further
restrictions of \IRS where respectively the set of guards and the set
of clocks of the controller are fixed a priori. Thus, we consider the
following hierarchy of problems:
$\RS\supseteq
\IRS\supseteq \begin{tabular}{c}{\BPRS}\\{\BCRS}\end{tabular}\supseteq
\BRRS$.
Unfortunately, while \IRS, \BPRS and \BCRS seem to make sense in
practice, they turn out to be undecidable both on finite and infinite
words---an answer to points~($i$) and~($ii$). Our proofs are based on
a \emph{novel} encoding of halting problem for deterministic channel
machines.  By contrast, the undecidability results of \cite{BouBoz06}
(for \MTL) are reductions from the same problem, but their encoding
relies heavily on the ability of \MTL to express \emph{punctual
  constraints} like `every $a$ event is followed by a $b$ event
\emph{exactly} one time unit later', which is not allowed by \MITL. To
the best of our knowledge, our proofs are the first to perform such a
reduction in a formalism that disallows punctual requirements---a
somewhat unexpected result. Then, we answer point~($iii$) by
considering a hierarchy of syntactic subsets of \MITL (see
\figurename~\ref{fig:decsummary}) and showing that, for all these
subsets, \BPRS and \BCRS (hence also \IRS and \RS) remain undecidable,
on finite and infinite words. Note that the undecidability proof of
\cite{DSoMad02} cannot easily be adapted to cope with these cases,
because it needs a mix of open and closed constraints; while we prove
undecidable very weak fragments of \MITL where only closed or only
open constraints are allowed.  All these negative results shape a
precise picture of the decidability border for real-time synthesis (in
particular, they answer open questions from
\cite{BouBoz06},\cite{BulDav14} and~\cite{DoyGee09}).  On the positive
side, we answer point ($iv$) by devising an on-the-fly algorithm to
solve \BRRS (in the finite words case) that runs in 3-$\EXP$. It
relies on one-clock alternating timed automata (as in~\cite{BouBoz06},
but unlike \cite{DSoMad02} that use timed automata), and on the
recently introduced \emph{interval semantics}~\cite{BriEst13}.

\section{Reactive synthesis of timed properties}

Let $\Sigma$ be a finite alphabet. A (finite) timed word\footnote{In
  order to keep the discussion focused and concise, we give the formal
  definitions for finite words only. It is straightforward to adapt
  them to the infinite words case.} over $\Sigma$ is a finite word
$\sigma = (\sigma_1,\tau_1)\cdots (\sigma_n,\tau_n)$ over
$\Sigma\times \R^+$ with $(\tau_i)_{1\leq i\leq n}$ a non-decreasing
sequence of non-negative real numbers. We denote by $T\Sigma^\star$
the set of finite timed words over $\Sigma$. A \emph{timed language}
is a subset $L$ of $T\Sigma^\star$.

\parag{Timed logics.} We consider the reactive synthesis problem
against various real-time logics, all of them being restrictions of
Metric Temporal Logic (\MTL)~\cite{Koy90}.  The logic \MTL is a timed
extension of \LTL, where the temporal modalities are labelled with a
timed interval. The formal syntax of \MTL is given as follows:
\begin{displaymath}
\phi := \top\ |\ a\ |\ \phi\wedge\phi\ |\ \neg\phi\ |\ \phi \Until[I]\phi
\end{displaymath}
where $a\in\Sigma$ and $I$ is an interval over $\mathbb{R}^+$ with endpoints
in $\N\cup\{+\infty\}$. 

We consider the \emph{pointwise semantics} and interpret \MTL formulas
over timed words. The semantics of a formula $\phi$ in \MTL is defined
inductively in the usual way. We recall only the semantics of
$\Until$: given
$\sigma=(\sigma_1,\tau_1)\cdots (\sigma_n,\tau_n)\in T\Sigma^\star$,
and a position $1\leq i\leq n$, we let
$(\sigma,i)\models\phi_1\Until[I] \phi_2$ if there exists $j>i$ such
that $(\sigma,j)\models\phi_2$, $\tau_j-\tau_i\in I$, and
$(\sigma, k)\models\phi_1$, for all $i< k< j$.

With $\bot := \neg\top$, we can recover the `next' operator
$\Next[I]\phi := \bot\Until[I] \phi$, and we rely on the usual
shortcuts for the `finally', `globally' and `dual-until' operators:
$\Finally[I]\phi:= \top\Until[I]\phi$,
$\Globally[I]\phi := \neg\Finally[I]\neg\phi$ and
$\phi_1\dualUntil[I] \phi_2 := \neg((\neg\phi_1)\Until[I](\neg
\phi_2))$.
We also use the non-strict version of the `until' operator
$\phi_1\wUntil[I] \phi_2$, defined as
$\phi_2\vee (\phi_1\wedge \phi_1\Until[I]\phi_2)$ (if $0 \in I$) or
$\phi_1\wedge \phi_1\Until[I]\phi_2$ (if $0 \notin I$).  This notation
yields the corresponding non-strict operators $\wFinally\phi$ and
$\wGlobally\phi$ in the natural way.
When the interval $I$ is the entire set of the non-negative real
numbers, the subscript is often omitted.  We say that $\sigma$
satisfies the formula $\phi$, written $\sigma\models\phi$ if
$(\sigma,1)\models\phi$, and we denote by $\Lang(\phi)$ the set of all
timed words $\sigma$ such that $\sigma\models\phi$.

We consider mainly a restriction of $\MTL$ called \MITL (for Metric
Interval Temporal Logic), in which the intervals are restricted to
non-singular ones. We denote by $\textsf{\textup{Open-}}\MITL$ the
open fragment of \MITL: in negation normal form, each subformula
$\phi_1 \Until[I] \phi_2$ has either $I$ open or $\inf(I) = 0$ and $I$
right-open, and each subformula $\phi_1 \dualUntil[I] \phi_2$ has $I$
closed. Then, a formula is in $\textsf{\textup{Closed-}}\MITL$ if it
is the negation of an $\textsf{\textup{Open-}}\MITL$
formula. 
By~\cite{BriEst13}, $\textsf{\textup{Open-}}\MITL$ formulas
(respectively, $\textsf{\textup{Closed-}}\MITL$ formulas) translate to
open (closed) timed automata~\cite{OuaWor03}, i.e., all clock constraints are strict
(non-strict). Two other important fragments of \MTL
considered in the literature consist of \SMTL~\cite{OuaWor07}, where
each subformula $\phi_1 \Until[I] \phi_2$ has $I$ bounded in {negation
  normal form}, and \CFMTL \cite{BouMar07}, where the formula
satisfies the following in negation normal form:
\begin{inparaenum}[(i)]
\item in each subformula $\phi_1 \Until[I] \phi_2$, if $I$ is
  unbounded then $\phi_2 \in \LTL$; and
\item in each subformula $\phi_1 \dualUntil[I] \phi_2$, if $I$ is
  unbounded then $\phi_1 \in \LTL$.
\end{inparaenum}

For all of these logics $\mathsf L$, we can consider several
restrictions.  The restriction in which only the non-strict variants
of the operators ($\wFinally$, $\wGlobally$, \textit{etc.}) are allowed is
denoted by $\mathsf L^{\mathrm{ns}}$. The fragment in which all the
intervals used in the formula are either unbounded, or have a left
endpoint equal to 0 is denoted by $\mathsf L[\Until[0,\infty]]$.  In
this case, the interval $I$ can be replaced by an expression of the
form $\sim c$, with $c\in\mathbb{N}$, and
${\sim}\in\{<,>,\leq,\geq\}$. It is known that
$\MITL[\Until[0,\infty]]$ is expressively equivalent to \ECLfut
\cite{Raskin1999}, which is itself a syntactic fragment of Event-Clock
Logic (\ECL). Finally, $\mathsf L[\Finally[\infty]]$ stands for the
logic where `until' operators only appear in the form of
$\Finally[I]$ or $\Globally[I]$ with intervals $I$ of the shape
$[a,\infty)$ or $(a,\infty)$.

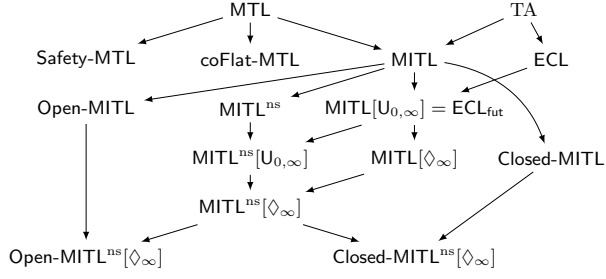
\begin{figure}[t]
\centering
\scalebox{.8}{\begin{tikzpicture}[->, >=latex, shorten >=1pt, transform shape,
  x=.9cm, y=.55cm, step=1]

\node		(ta) at (5, 0) {\TA};

\node		(mtl) at (0, 0)	{\MTL};

\node		(cfmtl) at (0, -1.5) {\CFMTL};

\node		(smtl) at (-3, -1.5) {\SMTL};

\node		(mitl) at (3, -1.5) {\MITL};

\node		(ecl) at (5.5, -1.5) {\ECL};

\node		(wmitl) at (0, -3) {$\MITL^{\mathrm{ns}}$};

\node		(eclfut) at (3, -3) {$\MITL[\Until[0, \infty]] = \ECLfut$};

\node		(weclfut) at (0, -4.5) {$\MITL^{\mathrm{ns}}[\Until[0, \infty]]$};

\node		(umitll) at (3, -4.5) {$\MITL[\Finally[\infty]]$};

\node		(wumitll) at (0, -6) {$\MITL^{\mathrm{ns}}[\Finally[\infty]]$};

\node		(openwumitll) at (-3, -7.5) {$\textsf{Open-}\MITL^{\mathrm{ns}}[\Finally[\infty]]$};
\node		(closedwumitll) at (3, -7.5) {$\textsf{Closed-}\MITL^{\mathrm{ns}}[\Finally[\infty]]$};

\node           (closedmitl) at (5.5,-4.5) {$\textsf{Closed-}\MITL$};
\node           (openmitl) at (-3,-3) {$\textsf{Open-}\MITL$};

\path		(mtl) edge [->] (cfmtl)
			(mtl) edge [->] (smtl)
			(mtl) edge [->] (mitl)
			(ta) edge [->] (mitl)
			(ta) edge [->] (ecl)
			(mitl) edge [->] (wmitl)
			(mitl) edge [->] (eclfut)
			(ecl) edge [->] (eclfut)
			(wmitl) edge [->] (weclfut)
			(eclfut) edge [->] (weclfut)
			(eclfut) edge [->] (umitll)
			(weclfut) edge [->] (wumitll)
			(umitll) edge [->] (wumitll)
			(wumitll) edge [->] (openwumitll)
			(wumitll) edge [->] (closedwumitll)
                        (mitl) edge[->, bend left] (closedmitl)
                        (mitl) edge[->] (openmitl)
                        (closedmitl) edge[->] (closedwumitll)
                        (openmitl) edge[->] (openwumitll)
                        ;

\end{tikzpicture}}
\caption{All the fragments of \MITL for which \BPRS and \BCRS are
  undecidable (hence also \RS and \IRS). $A\rightarrow B$ means that
  $A$ strictly contains $B$.}
\label{fig:decsummary}
\end{figure}

\parag{Symbolic transition systems.} Let $X$ be a finite set of
variables, called clocks. The set $\Guards(X)$ of \emph{clock
  constraints} $g$ over $X$ is defined by:
$g:= \top\mid g\land g \mid x\bowtie c$, where
${\bowtie}\in\{<,\leq,=,\geq,>\}$, $x\in X$ and $c\in\Q^+$.  A
\emph{valuation} over $X$ is a mapping $\nu\colon X\to \R^+$. The
satisfaction of a constraint $g$ by a valuation $\nu$ is defined in
the usual way and noted $\nu \models g$, and $\sem g$ is the set of
valuations $\nu$ satisfying $g$.  For $t\in\R^+$, we let $\nu+t$ be
the valuation defined by $(\nu+t)(x) = \nu(x)+t$ for all $x\in X$. For
$R\subseteq X$, we let $\nu[R\leftarrow 0]$ be the valuation defined
by $(\nu[R\leftarrow 0])(x) = 0$ if $x\in R$, and
$(\nu[R\leftarrow 0])(x) = \nu(x)$ otherwise.

Following the terminology of \cite{DSoMad02,BouBoz06}, a
\emph{granularity} is a triple $\mu = (X,m,K)$ where $X$ is a finite
set of clocks, $m\in\N\setminus\{0\}$, and $K\in\N$. A constraint $g$
is $\mu$-granular if $g\in\Guards(X)$ and each constant in $g$ is of
the form $\frac\alpha m$ with an integer $\alpha \leq K$.  A
\emph{symbolic alphabet} $\Gamma$ based on $(\Sigma,X)$ is a finite
subset of $\Sigma\times \AtomicGuards_{m,K}(X)\times 2^X$, where
$\AtomicGuards_{m,K}(X)$ denotes all atomic $(X,m,K)$-granular clock
constraints (i.e., clock constraints $g$ such that
$\sem g = \sem {g'}$ or $\sem g \cap \sem{g'} = \emptyset$, for every
$(X,m,K)$-granular clock constraint $g'$). Such a symbolic alphabet
$\Gamma$ is said $\mu$-granular. A \emph{symbolic word}
$\gamma = (\sigma_1,g_1,R_1) \cdots (\sigma_n,g_n,R_n)$ over $\Gamma$
generates a set of timed words over $\Sigma$, denoted by $tw(\gamma)$
such that $\sigma\in tw(\gamma)$ if
$\sigma = (\sigma_1,\tau_1)\cdots (\sigma_n,\tau_n)$, and there is a
sequence $(\nu_i)_{0\leq i\leq n}$ of valuations with $\nu_0$ the zero
valuation, and for all $1\leq i\leq n$,
$\nu_{i-1}+\tau_{i}-\tau_{i-1}\models g_{i}$ and
$\nu_i = (\nu_{i-1}+\tau_i-\tau_{i-1})[R_{i}\leftarrow 0]$ (assuming
$\tau_0=0$). Intuitively, each $(\sigma_i,g_i,R_i)$ means that action
$\sigma_i$ is performed, with guard $g_i$ satisfied and clocks in
$R_i$ reset.

A \emph{symbolic transition system} (STS) over a symbolic alphabet
$\Gamma$ based on $(\Sigma,X)$ is a tuple
$\T = (S,s_0,\transitions,S_f)$ where $S$ is a possibly infinite set
of locations, $s_0\in S$ is the initial location,
$\transitions \subseteq S\times \Gamma\times S$ is the transition
relation, and $S_f\subseteq S$ is a set of accepting locations
(omitted if all locations are accepting). An STS with finitely many
locations is a \emph{timed automaton} (\TA)~\cite{AluDil94}. For a
finite path
$\pi = s_1\xrightarrow{b_1} s_2 \xrightarrow{b_2}\cdots
\xrightarrow{b_{n}} s_{n+1}$
of $\T$ (i.e., such that $(s_i,b_i,s_{i+1})\in \transitions$ for all
$1\leq i\leq n$), the \emph{trace} of $\pi$ is the word
$b_1b_2\cdots b_{n}$, and $\pi$ is \emph{accepting} if
$s_{n+1}\in S_f$. We denote by $\Lang(\T)$ the language of $\T$,
defined as the timed words associated to symbolic words that are
traces of finite accepting paths starting in $s_0$.
We say that a timed action $(t,\sigma)\in\R^+\times \Sigma$ is
\emph{enabled} in $\T$ at a pair $(s,\nu)$, denoted by
$(t,\sigma)\in\Enabled_\T(s,\nu)$, if there exists a transition
$(s,(\sigma,g,R),s')\in\dtransitions$ such that $\nu+t\models g$. The
STS $\T$ is \emph{time-deterministic} if there are no distinct
transitions $(s,(\sigma,g_1,R_1),s_1)$ and $(s,(\sigma,g_2,R_2),s_2)$
in $\transitions$ and no valuation $\nu$ such that $\nu\models g_1$
and $\nu\models g_2$.
In a time-deterministic STS $\T = (S,s_0,\dtransitions,S_f)$,
for all timed words $\sigma$, there is at most
one path $\pi$ whose trace $\gamma$ verifies $\sigma\in
tw(\gamma)$.
In that case, we denote by $\dtransitions(s_0,\sigma)$ the unique (if
it exists) pair $(s,\nu)$ (where $s\in S$ and $\nu$ is a valuation)
reached after reading $\sigma\in tw(\gamma)$.

\begin{example}\label{ex:STS}
  A time-deterministic \TA $\Plant$ with a single clock $x$ is
  depicted in \figurename~\ref{fig:tdSTS-plant}.
  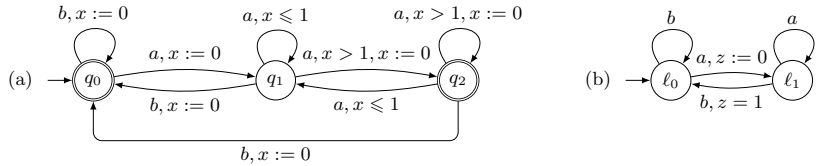
\begin{figure}[tbp]
    \centering
    \scalebox{.8}{\begin{tikzpicture}[->]
        \node[initial left,accepting,state] (0) {$q_0$};
        \node[state,right of=0] (1) {$q_1$};
        \node[accepting,state,right of=1] (2) {$q_2$};

        \path[draw] 
        (0) edge[loopabove] node[above] {$b, x:=0$} (0)
        (1) edge[loopabove] node[above] {$a,x\leq 1$} (1)
        (2) edge[loopabove] node[above] {$a,x> 1,x:=0$} (2)
        (0) edge[bend left=10] node[above] {$a,x:=0$} (1)
        (1) edge[bend left=10] node[below] {$b, x:=0$} (0)
        (1) edge[bend left=10] node[above] {$a,x>1,x:=0$} (2)
        (2) edge[bend left=10] node[below] {$a,x\leq 1$} (1);
        \draw[rounded corners] (2) -- (6,-1) -- node[below] {$b, x:=0$}
        (0,-1) -- (0);
        
        \node[left of=0,node distance=1.2cm] {(a)} ;

        \begin{scope}[node distance=2cm,xshift=9.5cm]
          \node[initial left,state] (0) {$\ell_0$};
          \node[state,right of=0] (1) {$\ell_1$};
          
          \path[draw] 
          (0) edge[loopabove] node[above] {$b$} (0)
          (1) edge[loopabove] node[above] {$a$} (1)
          (0) edge[bend left=10] node[above] {$a,z:=0$} (1)
          (1) edge[bend left=10] node[below] {$b,z=1$} (0);
          
          \node[left of=0,node distance=1.2cm] {(b)} ;
        \end{scope}

      \end{tikzpicture}}
    \caption{(a) A time-deterministic STS $\Plant$ with
      $X=\{x\}$. Instead of depicting a transition per letter
      $(a,g,R)$ (with $g$ atomic), we merge several transitions; e.g.,
      we skip the guard, when all the possible guards are
      admitted. $x:=0$ denotes the reset of $x$. (b) A
      time-deterministic STS $\T$. It is a controller to realise
      $\phi=\Globally(a\Rightarrow \Finally[\leq 1]b)$ with
      plant~$\Plant$.}
    \label{fig:tdSTS-plant}\label{fig:ex-controller}
  \end{figure}
  Intuitively, it accepts all timed words $\sigma$ of the form
  $w_1w_2\cdots w_n$ where each $w_i$ is a timed
  word such that
  \begin{inparaenum}[$(i)$]
  \item either $w_i=(b,\tau)$;
  \item or $w_i$ is a sequence of $a$'s (starting at time stamp
    $\tau$) of duration at most $1$; and $w_{i+1}$ is either of
    the form $(b,\tau')$, or of the form $(a,\tau')$ with $\tau'-\tau>1$.
  \end{inparaenum}
        
\end{example}

\parag{Reactive synthesis with plant.}
To define our reactive synthesis problems, we partition the alphabet
$\Sigma$ into controllable and environment actions $\Sigma_C$ and
$\Sigma_E$. Following \cite{DSoMad02,BouBoz06}, the system is modelled
by a time-deterministic \TA $\Plant = (Q,q_0,\ptransitions,Q_f)$,
called the \emph{plant}\footnote{We assume that for every location $q$ and every
  valuation $\nu$, there exists a timed action
  $(t,\sigma)\in\R^+\times \Sigma$ and a transition
  $(q,(\sigma,g,R),q')\in\ptransitions$ such that $\nu+t\models g$.}.
Observe that the plant has accepting locations: only those runs ending
in a final location of the plant will be checked against the
specification.
We start by recalling the definition of the general \emph{reactive
  synthesis} family of problems (\RS) \cite{AlfFae03,DoyGee09}.  It
consists in a game played by the controller and the environment, that
interact to create a timed word as follows. We start with the empty
timed word, and then, at each round, the controller and the
environment propose timed actions to be performed by the
system---therefore, they must be firable in the plant
$\Plant$---respectively $(t,a)$ and $(t',b)$, with $t,t'\in \R^+$,
$a\in\Sigma_C$ and $b\in\Sigma_E$.  The timed action with the
shortest\footnote{Observe that this is different from
  \cite{DSoMad02,BouBoz06}, where the environment can always prevent
  the controller from playing, even by proposing a longer delay. We
  claim our definition is more reasonable in practice but
  all proofs can be adapted to both definitions.} delay (or
the environment action if the controller decides not to propose any
action) is performed, and added to the current play for the next
round. If both players propose the same delay, we resolve the time
non-deterministically.

On those games, we consider a parameterised family of
reactive synthesis problems denoted $\RS_s^b(\mathcal{F})$, where
$s\in\{u,d\}$; $b\in\{\star,\omega\}$; and $\mathcal{F}$ is one of the
formalisms in \figurename~\ref{fig:decsummary}. An instance of
$\RS_s^b(\mathcal{F})$ is given by a specification $S\in\mathcal{F}$
and a plant $\Plant$, which are interpreted over finite words when
$b=\star$ and infinite words when $b=\omega$. The timed language
$\Lang(S)$ is a specification of desired behaviours when $s=d$ and
undesired behaviours when $s=u$. Then, $\RS_s^b(\mathcal{F})$ asks
whether there exists a strategy for the controller such that all the
words in the outcome of this strategy are in $\Lang(S)$ (or outside
$\Lang(S)$) when we consider desired (or undesired) behaviours (when
$s=\omega$, the definition of $\Lang(S)$ must be the infinite words
one).  If this is the case, we say that $S$ is \emph{(finite-word)
  realisable} for the problem under study.  For example,
$\RS_u^\omega(\MITL)$ is the reactive synthesis problem where the
inputs are a formula of $\MITL$ and a plant, which are interpreted
over the infinite words semantics, and where the $\MITL$ formula
specifies the behaviours that the controller should avoid. 
Unfortunately, the variants \RS are too general, and a winning strategy
might require unbounded memory:
\begin{example}\label{ex:RS}
  Consider the alphabet $\Sigma = \Sigma_C \uplus \Sigma_E$ with
  $\Sigma_C = \{b\}$ and $\Sigma_E = \{a\}$, a plant $\Plant$
  accepting $T\Sigma^\star$, and the specification defined by the \MTL
  formula $\phi = \Globally \big((a\land \Finally[\geq 1] a) \Rightarrow
  \Finally[=1] b\big)$.
  Clearly, a winning strategy for the controller is to remember the
  time stamps $\tau_1,\tau_2,\ldots$ of all $a$'s, and always propose
  to play action~$b$ one time unit later (note that if the environment
  blocks the time to prevent the controller from playing its $b$, the
  controller wins). However this requires to memorise an unbounded
  number of time stamps with a great precision.
\end{example}

\parag{Restrictions on the \RS problem.}
In practice, it makes more sense to restrict the winning strategy of
the controller to be implementable by an STS, which has finitely many
clocks (and if possible finitely many locations). Let us define
formally what it means for an STS $\T = (S,s_0,\dtransitions)$ to
control a plant $\Plant$. We
let $T\Sigma^\star_{\T,\Plant}$ be the \emph{set of timed words
  consistent with $\T$ and $\Plant$}, defined as the smallest set
containing the empty timed word, and closed by the following
operations. Let $\sigma$ be a word in $T\Sigma^\star_{\T,\Plant}$,
with $(q,\nu_\Plant)=\ptransitions(q_0,\sigma)$, $T=0$ if
$\sigma=\varepsilon$, and $(c,T)\in \Sigma\times \R^+$ be the last
letter of $\sigma$ otherwise. Then, we extend $\sigma$ as follows:
\begin{itemize}
\item either the controller proposes to play a controllable action
  $(t,b)$, because it corresponds to a transition that is firable both
  in the controller and the plant. This action can be played ($\sigma$
  is extended by $(b,T+t)$), as well as any environment action
  $(t',a)$ with $t'\leq t$ (the environment can overtake the
  controller). Formally, if $\dtransitions(s_0,\sigma)=(s,\nu)$ is
  defined and
  $\Enabled_\T(s,\nu)\cap \Enabled_\Plant(q,\nu_\Plant)\cap
  (\R^+\times \Sigma_C)\neq \emptyset$:
  for all
  $(t,b)\in \Enabled_\T(s,\nu)\cap\Enabled_\Plant(q,\nu_\Plant)\cap
  (\R^+\times \Sigma_C)$,
  we let $\sigma\cdot (b,T+t)\in T\Sigma^\star_{\T,\Plant}$ and
  $\sigma\cdot (a,T+t')\in T\Sigma^\star_{\T,\Plant}$ for all
  $t'\leq t$ and $a\in\Sigma_E$ such that
  $(t',a)\in\Enabled_\Plant(q,\nu_\Plant)$.
\item Or the controller proposes nothing, then the environment can
  play all its enabled actions. Formally, if
  $\dtransitions(s_0,\sigma)=(s,\nu)$ is defined and
  $\Enabled_\T(s,\nu)\cap \Enabled_\Plant(q,\nu_\Plant)\cap
  (\R^+\times \Sigma_C)=\emptyset$
  and
  $\Enabled_\Plant(q,\nu_\Plant)\cap (\R^+\times \Sigma_E)\neq
  \emptyset$,
  we let $\sigma\cdot (a,T+t')\in T\Sigma^\star_{\T,\Plant}$ for all
  $(t',a)\in \Enabled_\Plant(q,\nu_\Plant)\cap (\R^+\times \Sigma_E)$.
\item Otherwise, we declare that every possible future allowed by the
  plant is valid, i.e., we let
  $\sigma\cdot \sigma'\in T\Sigma^\star_{\T,\Plant}$ for all
  $\sigma\cdot\sigma'\in\Lang(\Plant)$. This happens when the
  controller proposes only actions that are not permitted by the plant
  while the environment has no enabled actions; or when the controller
  lost track of a move of the environment during the past.
\end{itemize}
Then, the \emph{\MTL implementable reactive synthesis} problem
$\IRS^\star_d(\MTL)$ (on finite words and with desired behaviours) is
to decide, given a plant $\Plant$ and a specification given as an \MTL
formula $\phi$, whether there exists a set of clocks $X$, a symbolic
alphabet $\Gamma$ based on $(\Sigma,X)$, and a time-deterministic STS
$\T$ over $\Gamma$ such that
$T\Sigma^\star_{\T,\Plant}\cap\Lang(\Plant)\subseteq
\Lang(\phi)\cup\{\varepsilon\}$.\footnote{Empty word $\varepsilon$ is added for
  convenience, in case it is not already in~$\Lang(\phi)$.}

While the definition of~$\IRS_d^\star(\MTL)$ is more practical than
that of~$\RS_d^\star(\MTL)$, it might still be too general because the
clocks and symbolic alphabet the controller can use are not fixed
\textit{a priori}. In the spirit of \cite{DSoMad02,BouBoz06}, we
define three variants of~\IRS. First, the \emph{ \MTL
  bounded-resources synthesis problem} $\BRRS_d^\star(\MTL)$ is a
restriction of $\IRS_d^\star(\MTL)$ where the granularity of the
controller is fixed: given an \MTL formula $\phi$, and a granularity
$\mu=(X,m,K)$, it asks whether there exists a $\mu$-granular symbolic
alphabet $\Gamma$ based on $(\Sigma,X)$, and a time-deterministic STS
$\T$ over $\Gamma$ such that
$T\Sigma^\star_{\T,\Plant}\cap\Lang(\Plant)\subseteq
\Lang(\phi)\cup\{\varepsilon\}$.
Second, the less restrictive \emph{\MTL bounded-precision synthesis
  problem} $\BPRS_d^\star(\MTL)$ and \emph{\MTL bounded-clocks
  synthesis problem} $\BCRS_d^\star(\MTL)$ are the variants of \IRS
where \emph{only} the precision and \emph{only} the number of clocks
are fixed, respectively. Formally, $\BPRS_d^\star(\MTL)$ asks, given
an \MTL formula $\phi$, $m\in \N$, and $K\in\N\setminus\{0\}$, whether
there are a finite set $X$ of clocks, an $(X,m,K)$-granular symbolic
alphabet $\Gamma$ based on $(\Sigma,X)$, and a time-deterministic STS
$\T$ over $\Gamma$ such that
$T\Sigma^\star_{\T,\Plant}\cap\Lang(\Plant)\subseteq
\Lang(\phi)\cup\{\varepsilon\}$.
$\BCRS_d^\star(\MTL)$ is defined similarly with an \MTL formula
$\phi$, and a finite set of clocks $X$ (instead of $m$, $K$) as input.

While we have defined \IRS, \BPRS, \BCRS and \BRRS for \MTL
requirements, and in the finite words, desired behaviours case only,
these definitions extend to all the other cases we have considered for
\RS: infinite words, undesired behaviours, and all fragments of
\MTL. We rely on the same notations as for \RS, writing for instance
$\BPRS_u^\star(\MITL)$ or $\BCRS_d^\omega(\CFMTL)$, etc.

\begin{example}
  Consider the instance of $\IRS_d^\star(\MITL)$ where the plant
  accepts $T\Sigma^\star$ and the specification is
  $\phi=\Globally(a\Rightarrow \Finally[\leq 1]b))$. This instance is
  negative ($\phi$ is not realisable), since, for every
  time-deterministic STS $\T$, $(a,0)\in T\Sigma^\star_{\T,\Plant}$
  but is not in $\Lang(\phi)$. However, if we consider now the plant
  $\Plant$ in \figurename~\ref{fig:tdSTS-plant}(a), we claim that the STS
  $\T$ with one clock $z$ depicted in
  \figurename~\ref{fig:ex-controller}(b) realises $\phi$. Indeed, this
  controller resets its clock $z$ each time it sees the first $a$ in a
  sequence of $a$'s, and proposes to play a $b$ when $z$ has value
  $1$, which ensures that all $a$'s read so far are followed by a $b$
  within $1$ time unit. The restrictions enforced by the plant (which
  can be regarded as a sort of fairness condition) ensure that this is
  sufficient to realise $\phi$ for $\IRS_d^\star(\MITL)$. This also
  means that $\phi$ is realisable for $\BPRS_d^\star(\MITL)$ with
  precision $m=1$ and $K=1$; for $\BCRS_d^\star(\MITL)$ with set of
  clocks $X=\{z\}$; and for $\BRRS_d^\star(\MITL)$ with granularity
  $\mu=(\{z\},1,1)$.
\end{example}

\section{\BPRS and \BCRS are undecidable}\label{sec:undecidability}
Let us show that all the variants of \BPRS and \BCRS are undecidable,
whatever formalism from \figurename~\ref{fig:decsummary} we consider
for the specification. This entails that all variants of \RS and \IRS
are undecidable too (in particular $\RS_d^\star(\ECL)$ which settles
an open question of~\cite{DoyGee09} negatively). To this aim, we show
undecidability on the weakest formalisms in
\figurename~\ref{fig:decsummary}, namely: \CFMTL, \SMTL,
$\textsf{Open-}\MITL^{\mathrm{ns}}[\Finally[\infty]]$ and
$\textsf{Closed-}\MITL^{\mathrm{ns}}[\Finally[\infty]]$.  Similar
results have been shown for \MTL (and for \SMTL as desired
specifications) in~\cite{BouBoz06} via a reduction from the halting
problem for deterministic channel machines, but their proof depends
crucially on \emph{punctual} formulas of the form
$\wGlobally (a \Rightarrow \wFinally[=1] b)$ which are not expressible
in \MITL.  Our original contribution here is to adapt these ideas to a
formalism without punctual constraints, which is non-trivial.

\parag{Deterministic channel machines.}
A \emph{deterministic channel machine} (DCM)
$\mathcal{S} = \langle S, s_0, s_\textit{halt}, M, \Delta \rangle$ can
be seen as a finite automaton equipped with an unbounded fifo channel,
where $S$ is a finite set of states, $s_0$ is the initial state,
$s_\textit{halt}$ is the halting state, $M$ is a finite set of
messages and
$\Delta \subseteq S \times \{m!, m? \mid m \in M \} \times S$ is the
transition relation satisfying the following \emph{determinism}
hypothesis:
\begin{inparaenum}[($i$)]
\item $(s, a, s') \in \Delta$ and $(s, a, s'') \in \Delta$ implies
  $s' = s''$;
\item if $(s, m!, s') \in \Delta$ then it is the only
  outgoing transition from $s$.
\end{inparaenum}

The semantics is described by a graph $G(\mathcal{S})$ with nodes
labelled by $(s, x)$ where $s \in S$ and $x \in M^\star$ is the channel
content.  The edges in $G(\mathcal{S})$ are defined as follows:
\begin{inparaenum}[($i$)]
\item $(s, x) \xrightarrow{m!} (s', xm)$ if $(s, m!, s') \in \Delta$; and
\item $(s, mx) \xrightarrow{m?} (s', x)$ if $(s, m?, s') \in \Delta$.
\end{inparaenum}
Intuitively, these correspond to messages being \emph{written to} or
\emph{read from} the channel.  A \emph{computation} of $\mathcal{S}$
is then a path in $G(\mathcal{S})$.  The \emph{halting problem} for
DCMs asks, given a DCM $\mathcal{S}$, whether there is a computation
from $(s_0, \varepsilon)$ to $(s_\textit{halt}, x)$ in $G(\mathcal{S})$
for some $x \in M^\star$.
\begin{proposition}[\cite{Brand83}]
The halting problem for DCMs is undecidable.
\end{proposition}
It should be clear that $\mathcal{S}$ has a unique computation.
Without loss of generality, we assume that
$s_\textit{halt}$ is the only state in $S$ with no outgoing transition.
It follows that exactly one of the following must be true:
\begin{inparaenum}[($i$)]
\item $\mathcal{S}$ has a halting computation; 
\item $\mathcal{S}$ has
  an infinite computation not reaching $s_\textit{halt}$;
\item $\mathcal{S}$ is blocking at some point, i.e., $\mathcal{S}$ is
  unable to proceed at some state $s \neq s_\textit{halt}$ (with only
  \emph{read} outgoing transitions) either because the channel is
  empty or the message at the head of the channel does not match any
  of the outgoing transitions from $s$.
\end{inparaenum}

\parag{Finite-word reactive synthesis for \MITL.}
We now give a reduction from the halting problem for DCMs to
$\RS^\star_d(\MITL)$.  The idea is to devise a suitable \MITL formula
such that in the corresponding timed game, the environment and the
controller are forced to propose actions in turn, according to the
semantics of the DCM.  Each prefix of the (unique) computation of the
DCM is thus encoded as a play, i.e., a finite timed word.  More
specifically, given a DCM $\mathcal{S}$, we require each play to
satisfy the following conditions:
\begin{itemize}
\item[\textsf{C1}] The action sequence of the play (i.e., omitting all
  timestamps) is of the form
  ${\textit{Nil}_\textit{C}}^\star\,s_0 a_0 s_1 a_1 \cdots$ where
  $\textit{Nil}_\textit{C}$ is a special action of the controller and
  $(s_i, a_i, s_{i+1}) \in \Delta$ for each $i \geq 0$.
\item[\textsf{C2}] Each $s_i$ comes with no delay and no two
  \emph{write} or \emph{read} actions occur at the same time, i.e., if
  $(a_i, \tau)(s_{i+1}, \tau')(a_{i+1}, \tau'')$ is a substring of the
  play then $\tau = \tau'$ and $\tau < \tau''$.
\item[\textsf{C3}] Each $m?$ is preceded exactly 1 time unit (t.u.)
  earlier by a corresponding $m!$.
\item[\textsf{C4}] Each $m!$ is followed exactly 1 t.u.\ later by a
  corresponding $m?$ if there are actions that occur at least 1
  t.u.\ after the $m!$ in question.
\end{itemize}
To this end, we construct a formula of the form
$\Phi \Rightarrow \Psi$ where $\Phi$ and $\Psi$ are conjunctions of
the conditions that the environment and the controller must adhere to,
respectively.  In particular, the environment must propose $s_i$'s
according to the transition relation (\textsf{C1} and \textsf{C2})
whereas the controller is responsible for proposing
$\{m!, m? \mid m \in M\}$ properly so that a correct encoding of the
writing and reading of messages is maintained (\textsf{C2},
\textsf{C3}, and \textsf{C4}).  When both players obey these
conditions, the play faithfully encodes a prefix of the computation of
$\mathcal{S}$, and the controller wins the play.  If the environment
attempts to ruin the encoding, the formula will be satisfied, i.e.,
the play will be winning for the controller.  Conversely, if the
controller attempts to cheat by, say, reading a message that is not at
the head of the channel, the environment can pinpoint this error (by
proposing a special action $\checkl$) and falsify the formula, i.e.,
the play will be losing for the controller.  In what follows, let
$\Sigma_\textit{E} = S \cup \{\textit{Check}^\leftarrow, \checkr,
\textit{Lose}, \textit{Nil}_\textit{E} \}$,
$\Sigma_\textit{C} = \{ m!, m? \mid m \in M \} \cup \{\textit{Win},
\textit{Nil}_\textit{C}\}$,
$\phi_\textit{E} = \bigvee_{e \in \Sigma_\textit{E}} e$,
$\phi_\textit{C} = \bigvee_{c \in \Sigma_\textit{C}} c$,
$\phi_S = \bigvee_{s \in S} s$,
$\phi_\textit{W} = \bigvee_{m \in M} m!$,
$\phi_\textit{R} = \bigvee_{m \in M} m?$ and
$\phi_\textit{WR} = \phi_\textit{W} \vee \phi_\textit{R}$. Let us now
present the formulas $\phi_1,\phi_2,\ldots$ and $\psi_1,\psi_2,\ldots$
needed to define $\Phi$ and $\Psi$.

We start by formulas enforcing condition \textsf{C1}. The play should
start from $s_0$, alternate between $E$-actions and $C$-actions, and
the controller can win the play if the environment does not proceed
promptly, and vice versa for the environment:
\begin{align*}
  \phi_1 &= \neg \big(\textit{Nil}_\textit{C} \wUntil (\phi_\textit{E} \wedge \neg s_0)\big) 
  & \psi_1 & = \neg \big(\textit{Nil}_\textit{C} \wUntil
             (\phi_\textit{C} \wedge \neg
             \textit{Nil}_\textit{C})\big) \\
  \phi_2 & = \neg \wFinally (\phi_\textit{E} \wedge \Next[\leq 1]
           \phi_\textit{E})  
  &  \psi_2 & = \neg \wFinally (\phi_\textit{C} \wedge \Next[\leq 1]
              \phi_\textit{C}) \\
  \phi_3 & = \neg \wFinally (\phi_\textit{WR} \wedge \Next \textit{Win}) 
  & \psi_3 & = \neg \wFinally (\phi_S \wedge \neg s_\textit{halt}
             \wedge \Next \textit{Lose}) \,. 
\end{align*}
Both players must also comply to the semantics of $\mathcal{S}$:
\begin{align*}
  \phi_4 & = \hspace{-5mm} \bigwedge_{\substack{(s, a, s') \in \Delta \\ b \notin \{s', \textit{Check}^\leftarrow, \checkr\}}} \hspace{-5mm}\neg \wFinally (s \wedge \Next a \wedge \Next \Next b) \hspace{5mm}
         & \psi_4 & = \hspace{-5mm}\bigwedge_{\substack{s \neq s_\textit{halt} \\
  \forall s'\, (s, a, s') \notin \Delta}} \hspace{-5mm} \neg \wFinally ( s \wedge \Next a ) \,.
\end{align*}
Once the encoding has ended, both players can only propose
$\textit{Nil}$ actions:
\begin{IEEEeqnarray*}{rCl}
\phi_5 & = & \neg \wFinally \big((s_\textit{halt} \vee \checkl \vee  \textit{Check}^\rightarrow \vee \textit{Lose} \vee  \textit{Win}) \wedge \Finally (\phi_\textit{E} \wedge \neg \textit{Nil}_\textit{E}) \big) \\
\psi_5 & = & \neg \wFinally \big((s_\textit{halt} \vee \checkl \vee  \textit{Check}^\rightarrow \vee \textit{Lose} \vee \textit{Win}) \wedge \Finally (\phi_\textit{C} \wedge \neg \textit{Nil}_\textit{C}) \big) \,.
\end{IEEEeqnarray*}

For condition \textsf{C2}, we simply state that the environment can
only propose delay $0$ whereas the controller always proposes a
positive delay:
\begin{align*}
  \phi_6 & =  \neg \wFinally (\phi_\textit{WR} \wedge \Next[>0]
           \phi_\textit{E}) \quad
  &
    \psi_6 & =  \wGlobally (\phi_S \wedge \neg s_\textit{halt} \wedge
             \Next \phi_\textit{WR} \implies \Next[>0]
             \phi_\textit{WR}) \,. 
\end{align*}

Let us finally introduce formulae to enforce conditions \textsf{C3}
and \textsf{C4}. Note that a requirement like `every write is matched
by a read \emph{exactly} one time unit later' is easy to express in
\MTL, but not so in \MITL. Nevertheless, we manage to translate
\textsf{C3} and \textsf{C4} in \MITL by exploiting the game
interaction between the players. Intuitively, we allow the cheating
player to be punished by the other. Formally, to ensure \textsf{C3},
we allow the environment to play a $\checkl$ action after any $m?$ to
check that this read has indeed occurred 1 t.u.\ after the
corresponding $m!$. Assuming such a $\checkl$ has occurred, the
controller must enforce:
\begin{IEEEeqnarray*}{rCl}
\psi^\leftarrow & = & \bigvee_{m \in M} \wFinally \big( m! \wedge \wFinally[\leq 1] (m? \wedge \Next \checkl)
 \wedge \wFinally[\geq 1] (m? \wedge \Next \checkl) \big) \,.
\end{IEEEeqnarray*}
Now, to ensure \textsf{C4}, the environment may play a $\checkr$
action at least 1 t.u.\ after a write on the channel. If this
$\checkr$ is the first action that occurs more than 1~t.u.\ after the
writing (expressed by the formula $\psifst$), we must check that the
writing has been correctly addressed, i.e., there has been an action
exactly 1~t.u.\ after, \emph{and} this action was the corresponding
reading:
\begin{IEEEeqnarray*}{rCl}
\psifst & = & \wFinally
( \phi_\textit{W} \wedge \wFinally[<1] \theta_{1}^\rightarrow
\wedge \wFinally[\geq 1] \theta_{0}^\rightarrow ) \\
\psi^\rightarrow & = & \neg \wFinally ( \phi_\textit{W} \wedge \wFinally[<1] \theta_{1}^\rightarrow
\wedge \wFinally[> 1] \theta_{0}^\rightarrow ) \wedge \psi^\leftarrow[\checkr/\checkl] 
\end{IEEEeqnarray*}
where $\psi^\leftarrow[\checkr/\checkl]$ is the formula obtained by
replacing all $\checkl$ with $\checkr$ in $\psi^\leftarrow$,
$\theta^\rightarrow_0 = \phi_\textit{WR} \wedge \Next \checkr$ and
$\theta^\rightarrow_1 = \phi_\textit{WR} \wedge \Next \phi_S \wedge \Next \Next \theta^\rightarrow_0$.
In the overall, we consider:
\begin{IEEEeqnarray*}{rCl}
\phi_7 & = & \bigwedge_{m \in M} \neg \wFinally (m! \wedge \Next \textit{Check}^\leftarrow) \\
\psi_7 & = & (\wFinally \checkl \Rightarrow \psi^\leftarrow) \wedge \big((\wFinally \textit{Check}^\rightarrow \wedge \psifst) \Rightarrow \psi^\rightarrow\big) \,.
\end{IEEEeqnarray*}

Now let $\Phi = \bigwedge_{1 \leq i \leq 7} \phi_i$,  $\Psi = \bigwedge_{1 \leq i \leq 7} \psi_i$
and $\Omega = \Phi \Rightarrow \Psi$.

\begin{proposition}\label{prop:realizability}
  $\Omega$ is finite-word realisable if and only if either ($i$)
  $\mathcal{S}$ has a halting computation, or ($ii$) $\mathcal{S}$ has
  an infinite computation not reaching
  $s_\textit{halt}$.\footnote{Observe that the proof does not require
    any plant (or uses the trivial plant accepting $T\Sigma^\star$). This
    entails undecidability of the `realisability problem', which is
    more restrictive than $\RS_d^\star$ and another difference with
    respect to the proof in \cite{BouBoz06}.}
\end{proposition}

\begin{proof}[Sketch]
  If ($i$) or ($ii$) is true, $\Omega$ can be realised by the controller
  faithfully encoding a computation of $\mathcal{S}$.  If $E$ proposes
  $\checkl$ or $\checkr$, the play will satisfy $\psi_7$. Otherwise,
  if $\mathcal{S}$ has an infinite computation not reaching
  $s_\textit{halt}$, the play can grow unboundedly and will satisfy
  all $\psi$'s, hence $\Omega$.

  Conversely, if $\mathcal{S}$ is blocking, then $\Omega$ is not
  realisable. Indeed, either the controller encodes $\mathcal{S}$
  correctly, but then at some point it will not be able to propose any
  action, and will be subsumed by the environment that will play
  $\textit{Lose}$.  Or the controller will try to cheat, by (1)
  inserting an action $m?$ not matched by a corresponding $m!$ 1 t.u.\
  earlier, or (2) writing a message $m!$ that will not be read 1 t.u.\
  later. For the first case, the environment can then play $\checkl$
  right after the incorrect $m?$, and the play will violate
  $\psi^\leftarrow$, hence $\psi_7$ and $\Omega$. For the second case,
  the environment will play $\checkr$ after the first action occurring
  1 t.u.\ after the unfaithful $m!$ and the play will violate
  $\psi^\rightarrow$.\qed
\end{proof}

Now let
$\Omega' = \Phi \Rightarrow \Psi \wedge \Globally (\neg
s_\textit{halt})$,
i.e., we further require the computation not to reach
$s_\textit{halt}$. The following proposition can be proved almost
identically.
\begin{proposition}
  $\Omega'$ is finite-word realisable if and only if $\mathcal{S}$ has
  an infinite computation not reaching $s_\textit{halt}$.
\end{proposition}
\begin{corollary}\label{cor:halting}
  $\mathcal{S}$ has a halting computation if and only if $\Omega$ is
  finite-word realisable but $\Omega'$ is not finite-word realisable.
\end{corollary}
It follows that if $\RS^\star_d(\MITL)$ is decidable, we can decide
whether $\mathcal{S}$ has a halting computation. But the latter is
known to be undecidable.  Hence:
\begin{theorem}\label{thm:realizability}
$\RS^\star_d(\MITL)$ is undecidable.
\end{theorem}
Theorem~\ref{thm:realizability} and its proof are the core results
from which we will derive all other undecidability results announced
at the beginning of the section.
\begin{remark}\label{rem:smtl}
  One may show that the $\RS^\omega_d$ problem is undecidable for
  formulas of the form $\Phi \Rightarrow \Psi$ where $\Phi$ and $\Psi$
  are conjunctions of formulas in $\SMTL[\Until[0,\infty]]$ by
  rewriting $\phi_i$'s and $\psi_i$'s (see
  Appendix~\ref{app:smtlrealisability}). This answers an open question
  of~\cite{BulDav14}. 
\end{remark}

\parag{\BPRS and \BCRS for \SMTL, \CFMTL, and \MITL.}
In the proof of Proposition~\ref{prop:realizability}, if $\mathcal{S}$
actually halts, the number of messages present in the channel during
the (unique) computation is bounded by a number $N$.  It follows that
the strategy of $C$ can be implemented as a bounded-precision
controller (with precision $(m, K) = (1, 1)$ and $N$ clocks) or a
bounded-clocks controller (with precision $(m, K) = (\frac{1}{N}, 1)$
and a single clock).  
Corollary~\ref{cor:halting} therefore holds also for the
bounded-precision and bounded-clocks cases, and the
$\BPRS^\star_d(\MITL)$ and $\BCRS^\star_d(\MITL)$ problems are
undecidable. By further modifying the formulas used in the proof of
Proposition~\ref{prop:realizability}, we show that the undecidability
indeed holds even when we allow only unary non-strict modalities with
lower-bound constraints and require the constraints to be exclusively
strict or non-strict (see Appendix~\ref{app:openclosedmitl}), hence
$\BPRS^\star_d$ and $\BCRS^\star_d$ are undecidable too on
$\textsf{Open-}\MITL^{\mathrm{ns}}[\Finally[\infty]]$ and
$\textsf{Closed-}\MITL^{\mathrm{ns}}[\Finally[\infty]]$. This entails
undecidability in the \emph{undesired specifications} case because the
negation of an $\textsf{Open-}\MITL^{\mathrm{ns}}[\Finally[\infty]]$
is a $\textsf{Closed-}\MITL^{\mathrm{ns}}[\Finally[\infty]]$ formula
and vice-versa.  Finally, we can extend our proofs to the infinite
words case (see Appendix~\ref{app:openclosedmitl}), hence:
\begin{theorem}\label{thm:openclosedmitl}
  $\RS^b_s(\mathsf{L})$, $\IRS^b_s(\mathsf{L})$,
  $\BPRS^b_s(\mathsf{L})$ and $\BCRS^b_s(\mathsf{L})$ are
  undecidable for 
  $\raisebox{-0.8pt}{$\mathsf{L}$}\in\{\mathsf{Open-}\MITL^{\mathrm{ns}}[\Finally[\infty]],
  \mathsf{Closed-}\MITL^{\mathrm{ns}}[\Finally[\infty]]\}$,
  $s\in\{u,d\}$ and $b\in \{\star,\omega\}$.
\end{theorem}
This result extends the previous undecidability proofs of
\cite{DoyGee09} ($\RS^\omega_d(\ECL)$ is undecidable), and
of~\cite{DSoMad02} ($\IRS^\star_d(\TA)$ and $\IRS^\star_u(\TA)$ are
undecidable). In light of these previous works, our result is somewhat
surprising as the undecidability proof in~\cite{DSoMad02} is via a
reduction from the universality problem for timed automata, yet this
universality problem becomes decidable when all constraints are
strict~\cite{OuaWor03}.

Finally, it remains to handle the cases of \SMTL and \CFMTL. Contrary
to the case of \MTL, the infinite-word satisfiability problem is
decidable for \SMTL~\cite{OuaWor07} and the
infinite-word model-checking problem is decidable for both
\SMTL~\cite{OuaWor07} and \CFMTL~\cite{BouMar07}.  Nevertheless, our
synthesis problems remain undecidable for these fragments (see
Appendix~\ref{app:smtlcfmtl}).  In particular, the result on \SMTL answers an open
question of~\cite{BouBoz06} negatively:
\begin{theorem}\label{thm:smtlcfmtl}
  $\RS^b_s(\mathsf{L})$, $\IRS^b_s(\mathsf{L})$,
  $\BPRS^b_s(\mathsf{L})$ and $\BCRS^b_s(\mathsf{L})$ are
  undecidable for $\raisebox{-0.8pt}{$\mathsf{L}$}\in\{\SMTL, \CFMTL\}$, $s\in\{u,d\}$ and
  $b\in \{\star,\omega\}$.
\end{theorem}

\section{Bounded-resources synthesis for \MITL
  properties}\label{sec:BoundReacSynt} We have now characterised
rather precisely the decidability border for \MITL synthesis
problems. In light of these results, we focus now on
$\BRRS^\star_d(\MITL)$ (since \MITL is closed under complement, one
can derive an algorithm for $\BRRS^\star_u(\MITL)$ from our
solution). Recall that the algorithm of D'Souza and Madhusudan
\cite{DSoMad02}, associated with the translation of \MITL into \TA
\cite{AluFed96} yields a 3\EXP\ procedure for these two
problems. Unfortunately this procedure is unlikely to be amenable to
efficient implementation. This is due to the translation from \MITL to
\TA and the need to determinise a region automaton, which is known to
be hard in practice. On the other hand, Bouyer \textit{et al.}
\cite{BouBoz06} present a procedure for $\BRRS^\star_d(\MTL)$ (which
can thus be applied to \MITL requirements). This algorithm is
on-the-fly, in the sense that it avoids, if possible to build a full
automaton for the requirement; and thus more likely to perform well in
practice. Unfortunately, being designed for \MTL, its running time can
only be bounded above by a non-primitive recursive function.  We
present now an algorithm for $\BRRS^\star_d(\MITL)$ that combines the
advantages of these two previous solutions: it is \emph{on-the-fly}
and runs in 3\EXP. To obtain an on-the-fly algorithm, Bouyer
\textit{et al.} use \emph{one-clock alternating automata} (\OCATA)
instead of \TA to represent the \MITL requirement. We follow the same
path, but rely on the newly introduced \emph{interval-based
  semantics}~\cite{BriEst13} for these automata, in order to mitigate
the complexity. Let us now briefly recall these two basic ingredients
(due to lack of space, we only sketch the algorithm, a more complete
presentation is in Appendix~\ref{sec:3exp-algor-brrsst}).

\parag{\OCATA\ and interval semantics.} Alternating timed automata
\cite{OuaWor07} extend (non-deterministic) timed automata by adding
\emph{conjunctive transitions}. Intuitively, conjunctive transitions
spawn several copies of the automaton that run in parallel from the
target states of the transition. A word is accepted iff \emph{all}
copies accept it. An example is shown in \figurename~\ref{fig:OCATA},
where the conjunctive transition is the hyperedge starting from
$\ell_0$. In the classical semantics, an execution of an \OCATA\ is a
sequence of set of states, named \emph{configurations}, describing the
current location and clock valuation of all active copies. For
example, a prefix of execution of the automaton in
\figurename~\ref{fig:OCATA} would start in $\{(\ell_0,0)\}$
(initially, there is only one copy in $\ell_0$ with the clock equal to
$0$); then $\{(\ell_0,0.42)\}$ (after letting $0.42$ time units
elapse); then $\{(\ell_0,0.42),(\ell_1,0)\}$ (after firing the
conjunctive transition from $\ell_0$), etc. It is well-known that all
formulas $\phi$ of \MTL (hence, also \MITL) can be translated into an
\OCATA\ $A_\phi$ that accepts the same language \cite{OuaWor07} (with
the classical semantics); and with a number of locations linear in the
number of subformulas of $\phi$. This translation is thus
straightforward. This is the key advantage of \OCATA\ over \TA: the
complexity of the \MITL formula is shifted from the syntax to the
semantics---what we need for an on-the-fly algorithm.
\begin{figure}[tbp]
  \centering
  \scalebox{.8}{\begin{tikzpicture}[node distance = 2.5cm]
    \node[initial left,accepting,state](0){$\ell_0$};
    \node[circle, minimum size = 0.5em, inner sep = 0mm, fill, right of=0](1){};
    \node[state, right of=1](2){$\ell_1$};
    
    \path (0) edge[loopabove,->] node[above]{$b$} (0)
    (0) edge node[above]{$a$} (1)
    (1) edge[->,bend left] node[below] {} (0)
    (1) edge[->] node[above] {$y:=0$} (2)
    (2) edge[->,loopabove] node[above]{$a$} (2)
    (2) edge[->] node[above]{$y\leq 1,b$} (7,0);
  \end{tikzpicture}}
  \caption{An \OCATA\ (with single clock $y$) accepting the language of
    $\Globally(a\Rightarrow \Finally[\leq 1]b)$.}
  \label{fig:OCATA}
\end{figure}
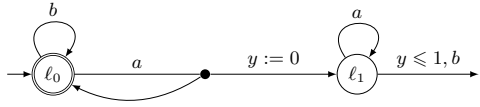

Then; in the \emph{interval semantics} \cite{BriEst13}, valuations of
the clocks are not \emph{points} anymore but
\emph{intervals}. Intuitively, intervals are meant to approximate sets
of (punctual) valuations: $(\ell,[a,b])$ means that there \emph{are}
clock copies with valuations $a$ and $b$ in $\ell$, and that there
\emph{could be} more copies in $\ell$ with valuations in $[a,b]$. In
this semantics, we can also \emph{merge} two copies
$(\ell, [a_1,b_1])$ and $(\ell, [a_2,b_2])$ into a single copy
$(\ell, [a_1,b_2])$ (assuming $a_1\leq b_2$), in order to keep the
number of clock copies below a fixed threshold $K$. It has been shown
\cite{BriEst13} that, when the \OCATA has been built from an \MITL
formula, the interval semantics is sufficient to retain the language
of the formula, with a number of copies which is at most doubly
exponential in the size of the formula.

\parag{Sketch of the algorithm.} Equipped with these elements, we can
now sketch our algorithm for $\BRRS_d^\star(\MITL)$. Starting from an
\MITL formula $\phi$, a plant $\Plant$ and a granularity
$\mu = (X,m,K)$, we first build, in polynomial time, an \OCATA
$A_{\neg\phi}$ accepting $\Lang(\neg\phi)$. Then, we essentially adapt
the technique of Bouyer \textit{et al.} \cite{BouBoz06}, relying on
the interval semantics of \OCATA instead of the classical one. This
boils down to building a tree that unfolds the parallel execution of
$A_{\neg\phi}$ (in the interval semantics), $\Plant$ and all possible
actions of a $\mu$-granular controller (hence the \emph{on-the-fly}
algorithm). Since the granularity is fixed, there are only finitely
many possible actions (i.e., guards and resets on the controller
clocks) for the controller at each step. We rely on the region
construction to group the infinitely many possible valuations of the
clocks into finitely many equivalence classes that are represented
using `region words' \cite{OuaWor07}. The result is a finitely
branching tree that might still have infinite branches. We stop
developing a branch once a global configuration (of $A_{\neg\phi}$,
$\Plant$, and the controller) repeats on the branch. By the region
construction \emph{and} the interval semantics, this will happen on
all branches, and we obtain a \emph{finite tree} of size at most
triply exponential. This tree can be analysed (using backward
induction) as a game with a safety objective for the controller: to
avoid the nodes where $\Plant$ and $A_{\neg\phi}$ accept at the same
time. The winning strategy yields, if it exists, a correct controller.

\begin{table}[tbp]
  \caption{Experimental results on the scheduling problem: realisable
    instances on the left, non-realisable on the right.}\label{tab:scheduler}
\centering{\scriptsize
 \begin{tabular}{|c|c|c|c|}
   \hline
   $T$ & $n$ & \# clocks & exec. time (sec) / \#nodes \\
  \hline
   1 & 1 & 0 & 46 / 52  \\
   1& 1& 1  & 199 / 147  \\
   1& 1& 2 & 4,599 / 1,343 \\
   
   2 &2 & 1 & 2,632 / 645 \\
   2&2 & 2 & 18,453 / 2,358  \\

   3 &3 & 1 & 182,524 / 2,297 \\
   3 &3 & 2 & $>$5min  \\
  
   4 & 4 & 0 & 54,893 / 667 \\
   4&4 & 1 & $>$5min \\
   \hline
   \end{tabular}
   \hspace*{.5cm}
   \begin{tabular}{|c|c|c|c|}
     \hline
     $T$ & $n$ & \# clocks & exec. time (sec) / \#nodes \\
     \hline
     2 & 1 & 0 & 77 / 84  \\
     2&1 & 1 & 824 / 311\\
     2&1 & 2 & 3,079 / 1,116  \\
     
     3 &2 & 1 & 17,134 / 1698 \\
     3 &2 & 2 & $>$5min \\
     
     4 & 3 & 0 & 10,621 / 540  \\
     4&3 & 1 & $>$5min  \\
     \hline
   \end{tabular}
 }
\end{table}

\parag{Experimental results.} We have implemented our procedure in
Java, and tested it over a benchmark related to a scheduling problem,
inspired by an example of~\cite{BulDav14}. This problem considers $n$
machines, and a list of jobs that must be assigned to the machines. A
job takes $T$ time units to finish. The plant ensures that at least
one time unit elapses between two job arrivals (which are
uncontrollable actions). The specification asks that the assignment be
performed in 1 time unit, and that each job has $T$ time units of
computation time. We tested this example with $T=n$, in which case the
specification is realisable (no matter the number of clocks, which we
make vary for testing the prototype efficiency), and with $T=n+1$, in
which case it is not.  Table~\ref{tab:scheduler} summarises some of
our results.

These results show that our prototypes can handle small but
non-trivial examples.  Unfortunately---as expected by the high
complexities of the algorithm---they do not scale well. As future
works, we will rely on the well-quasi orderings defined in
\cite{BouBoz06} to introduce heuristics in the spirit of the antichain
techniques \cite{FilJin09}. Second, we will investigate zone-based
versions of this algorithm to avoid the state explosion which is
inherent to region based techniques.


\newpage 
\appendix
\changepage{3cm}{3cm}{-1.5cm}{-1.5cm}{}{-1cm}{}{}{}

\section{Proof of Proposition~\ref{prop:realizability}}\label{app:realizability}

If ($i$) or ($ii$) is true, then $\Omega$ can be realized by the following strategy:

\begin{enumerate}
\item In round $0$, $C$ proposes $(\delta^0_\textit{C}, \textit{Nil}_\textit{C})$ with any $\delta^0_\textit{C} \in \mathbb{R}_{\geq 0}$. 
If $(\sigma_1, \tau_1) = (\textit{Nil}_\textit{C}, \delta^0_\textit{C})$ then
$C$ proposes $(\delta^1_\textit{C}, \textit{Nil}_\textit{C})$ with $\delta^1_\textit{C} > 1$ in round $1$;
if $(\sigma_2, \tau_2) = (\textit{Nil}_\textit{C}, \tau_1 + \delta^1_\textit{C})$ then
$C$ proposes $(\delta^2_\textit{C}, \textit{Nil}_\textit{C})$ with $\delta^2_\textit{C} > 1$
in round $2$, and so on.
If $E$ never supersedes then the play satisfies $\Omega$ (as it never violates any of $\psi$'s).
If at any point $E$ supersedes with an action other than $s_0$, the play will again satisfy $\Omega$ (by violating $\varphi_1$). 
If $E$ supersedes with $s_0$, since $\Psi$ is not violated, $C$ can proceed to step 2.
\item In round $i$ ($i > 0$) with $\sigma_{i} = s$ for some
$s \in S \setminus \{s_\textit{halt}\}$, $C$ proposes $(\delta^i_\textit{C}, a)$ with
$(s, a, s') \in \Delta$ for some $s'$ (this corresponds to a transition in the computation).
If the channel is not empty before round $i$, let $(m'!, \tau_j)$ ($j < i$) 
be the event that corresponds to the oldest pending message, i.e., the message at the head of the channel.
If $a = m!$ for some $m$, it must happen before the oldest pending message is read,
i.e., $\tau_{i} + \delta^i_\textit{C} < \tau_j + 1$;
if the channel is empty then let $\delta^i_\textit{C} \leq 1$.
If $a = m?$ for some $m$, it must be reading the oldest pending message, i.e., $\tau_{i} + \delta^i_\textit{C} = \tau_j + 1$.
Since we have $\delta^i_\textit{C} \leq 1$ in all these cases,
$\varphi_2$ will be violated if $E$ supersedes. If $E$ does not supersede, proceed to step 3.
\item In round $i$ ($i > 1$) with $\sigma_{i-1} = s$ and $\sigma_{i} = a$ for some $(s, a, s') \in \Delta$,
$C$ proposes $(\delta^i_\textit{C}, \textit{Win})$ with $\delta^i_\textit{C} > 1$.
If $E$ does not supersede, $\varphi_3$ will be violated.
If $E$ supersedes with some positive delay, $\varphi_6$ will be violated. 
If $E$ proposes $(0, b)$ with $b \notin \{s', \checkl, \checkr\}$, $\varphi_4$ will be violated.
Now consider the remaining cases:
	\begin{itemize}
	\item If $(\sigma_{i+1}, \tau_{i+1}) = (s', \tau_{i})$ and $s' \neq s_\textit{halt}$,  go back to step $2$.
	\item If $(\sigma_{i+1}, \tau_{i+1}) = (s', \tau_{i})$ and $s' = s_\textit{halt}$, proceed to step $4$.
	\item If $(\sigma_{i+1}, \tau_{i+1}) = (\checkl, \tau_{i})$, then $\varphi_7$ will be violated
			if $a = m!$ for some $m$. Otherwise if $a = m?$ for some $m$, since $\tau_{i}$ is exactly
			1 t.u. after the corresponding $m!$ (by step $2$), all $\psi$'s (in particular $\psi_7$) hold
			and $C$ may proceed to step $4$. 
	\item If $(\sigma_{i+1}, \tau_{i+1}) = (\checkr, \tau_{i})$, then if $a = m!$ for some $m$,
			either (1) $\tau_{i} < \tau_j + 1$ where $(m'!, \tau_j)$ corresponds
			to the oldest pending message (by step $2$) or (2) all `write' actions before this $m!$
			have been followed 1 t.u. later by a corresponding `read'.
			In both cases $\psifst$ does not hold, hence $\psi_7$ holds. If $a = m?$ for some $m$,
			then $\psi^\rightarrow$ and hence $\psi_7$ clearly holds (also by step $2$).
	\end{itemize}
\item Starting from round $i$ ($i > 0$) with $\sigma_{i} \in \{s_\textit{halt}, \textit{Check}^\leftarrow, \checkr\}$, $C$ proposes $(\delta^i_\textit{C}, \textit{Nil}_\textit{C})$ with
$\delta^i_\textit{C} > 1$ in the remaining rounds. If $E$ supersedes with an action other than
$\textit{Nil}_\textit{E}$, $\varphi_5$ will be violated. 
\end{enumerate}
Note that if $\mathcal{S}$ has an infinite computation not reaching $s_\textit{halt}$ and $E$ never proposes
$\checkl$ or $\checkr$, a play can grow unboundedly and $C$ will never reach step $4$; but any such play
will satisfy all $\psi$'s and hence $\Omega$. 
For the other direction, we show that if $\mathcal{S}$ is blocking at some point then $\Omega$ is not realizable,
i.e., $E$ can force a play that violates $\Omega$ for any strategy of $C$.
\begin{enumerate}
\item In round $0$, $E$ proposes $(0, s_0)$. By the rules of the timed game, it is clear that no matter what $C$ proposes,
there will be a play starting with $(s_0, 0)$, and $E$ can proceed to step $2$.
\item In round $i$ ($i > 0$) with $\sigma_{i} = s$ for some $s \in S \setminus \{ s_\textit{halt} \}$, 
$E$ proposes $(1.1, \textit{Lose})$. If $C$ does not supersede then $\psi_3$ will be violated.
If $C$ supersedes with a delay of $0$, $\psi_6$ will be violated.
If $C$ supersedes with an action other than the available transitions at $s$, $\psi_4$ will be violated.
It remains to consider the case $(\sigma_{i+1}, \tau_{i+1}) = (a, \tau_{i} + \delta^i_\textit{C})$
for some $(s, a, s') \in \Delta$ and $\delta^i_\textit{C} > 0$. Consider the following cases:
	\begin{itemize}
	\item If $a = m!$ for some $m$ and either (1) $\tau_{i+1} < \tau_j + 1$ where $(m'!, \tau_j)$ corresponds
			to the oldest pending message or (2) all `write' actions before this $m!$ have been followed 1 t.u. later
			by a corresponding `read',
			 proceed to step $3$; otherwise proceed to step $4$.
	\item If $a = m?$ for some $m$ and (1) $\tau_{i} + \delta^i_\textit{C} = \tau_j + 1$ for
	some $(m!, \tau_j)$ in the play and (2) $(m!, \tau_j)$ is the oldest pending message,
	proceed to step $3$; otherwise if $(m!, \tau_j)$ is not the oldest pending message, proceed to step $4$.
	If there is no such $(m!, \tau_j)$ in the play, proceed to step $5$.
	\end{itemize}
\item In round $i$ ($i > 0$) with $\sigma_{i} = a \in \{m!, m? \mid m \in M\}$, $E$ proposes
$(0, s')$ where $(s, a, s') \in \Delta$. There will be a play with $(\sigma_{i+1}, \tau_{i+1}) = (s', \tau_{i})$
and $E$ can go back to step $2$.
\item In round $i$ ($i > 0$) with $\sigma_{i} = a$ and either (1) $a = m!$ and $\tau_{i} \geq \tau_{j} + 1$
where $\tau_j$ is the timestamp of the oldest pending message or (2) $a = m?$ and $(m!, \tau_{k})$ with
$\tau_{i} = \tau_k + 1$ is in the play but is not the oldest pending message, $E$ proposes
$(0, \checkr)$; there will be a play with $(\sigma_{i+1}, \tau_{i+1}) = (\checkr, \tau_{i})$ (violating $\psi_7$).
\item In round $i$ ($i > 0$) with $\sigma_{i} = m?$ and there is no $(m!, \tau_j)$ in the play
with $\tau_{i} = \tau_j + 1$, $E$ proposes
$(0, \textit{Check}^\leftarrow)$; there will be a play with $(\sigma_{i+1}, \tau_{i+1}) = (\textit{Check}^\leftarrow, \tau_{i})$ (violating $\psi_7$).
\end{enumerate}
Either $E$ wins at step $2$ (if $C$ does not supersede) or $E$ will eventually proceed to step $4$ or $5$ and wins.

\section{Proof of Remark~\ref{rem:smtl}}\label{app:smtlrealisability}
It is clear that $\{\phi_i \mid 1 \leq i \leq 6 \}$ and  $\{ \psi_i \mid 1 \leq i \leq 5 \}$ are already 
$\SMTL[\Until[0,\infty]]$ formulas.
For $\psi_6$, we can replace it by $\neg \wFinally (\phi_S \wedge \neg s_\textit{halt} \wedge \Next[\leq 0] \phi_\textit{WR})$. For $\phi_7$ and $\psi_7$, we can replace $\psi^\leftarrow$ by
\begin{IEEEeqnarray*}{l}
\neg \wFinally[<1] \checkl \wedge \neg \wFinally \big(\wFinally[>1] \checkl \wedge \Next (\wFinally[<1] \checkl) \big) \\
\wedge \bigwedge_{\substack{m, m' \in M \\ m \neq m'}} \neg \wFinally \big((m! \vee m?) \wedge
\wFinally[\leq 1] (m'? \wedge \Next \checkl) \wedge \wFinally[\geq 1] (m'? \wedge \Next \checkl) \big)
\end{IEEEeqnarray*}
and replace $\psi^\leftarrow[\checkr/\checkl]$ accordingly.

\section{Proof of Theorem~\ref{thm:smtlcfmtl}}\label{app:smtlcfmtl}

We give reductions from the halting problem for DCMs to $\BPRS_u^b(\SMTL)$ and $\BPRS_d^b(\CFMTL)$
($b \in \{\star, \omega\}$), $\BPRS_u^\star(\CFMTL)$, and the corresponding \BCRS problems.
The remaining cases follow (more or less) directly from existing results~\cite{BouBoz06,BouMar07}.
The encoding we use here is very similar to the one used in Section~\ref{sec:undecidability}---the
main difference is that we now use a plant, in place of formulas, to ensure $\textsf{C1}'$ and $\textsf{C2}'$ (see below).
Without loss of generality, we consider a DCM $\mathcal{S} = \langle S, s_0, s_\textit{halt}, M, \Delta \rangle$
where $s_0 \neq s_\textit{halt}$, $s_0$ has an outgoing `write' action, and $s_\textit{halt}$ is the only state in $S$ with no outgoing transition. The plant $\mathcal{P} = \langle Q, q_0, \rightarrow, F \rangle$ 
over $(\Sigma_\textit{C} \cup \Sigma_\textit{E}, X)$ is constructed as follows:
\begin{itemize}
\item $\Sigma_\textit{C} = \{ m!, m? \mid m \in M \}$,  $\Sigma_\textit{E} = \{\textit{Check}^\leftarrow, \textit{Check}^\rightarrow, \textit{Nil}, \textit{Halt}, \textit{End} \}$, $X = \{x\}$;
\item $Q = S \cup \{q_\delta \mid \delta \in \Delta\}$, $q_0 = s_0$, $F = \{s_\textit{halt}\}$, and
	\begin{itemize}
	\item $s \xrightarrow{a, x > 0} q_\delta$ and $q_\delta \xrightarrow{\textit{Nil}, x = 0} s'$ for all $\delta = (s, a, s') \in \Delta$ with $s' \neq s_\textit{halt}$
	\item $s \xrightarrow{a, x > 0} q_\delta$ and $q_\delta \xrightarrow{\textit{Halt}, x = 0} s_\textit{halt}$ for all $\delta = (s, a, s_\textit{halt}) \in \Delta$
	\item $q_\delta \xrightarrow{\textit{Check}^\leftarrow, x = 0} s_\textit{halt}$ for all $\delta = (s, m?, s') \in \Delta$
	\item $q_\delta \xrightarrow{\textit{Check}^\rightarrow, x = 0} s_\textit{halt}$ for all $\delta \in \Delta$
	\item $s_\textit{halt} \xrightarrow{\textit{End}, \mathbf{true}} s_\textit{halt}$
	\end{itemize}
where $x$ is reset on every transition.
\end{itemize}
In what follows, let 
$\hat{\theta}^\rightarrow_0 = \varphi_\textit{WR} \wedge (\varphi_\textit{WR} \wUntil \checkr)$,
$\hat{\theta}^\rightarrow_1 = \varphi_\textit{WR} \wedge \Big(\varphi_\textit{WR} \wUntil \big(\textit{Nil} \wedge (\textit{Nil} \wUntil \hat{\theta}^\rightarrow_0) \big)\Big)$ and
$\hat{\theta}^\leftarrow_0 = \hat{\theta}^\rightarrow_0[\checkl/\checkr]$.
The encoding we have in mind is as follows:
\begin{itemize}
\item[$\textsf{C1}'$] The action sequence of the play (i.e. omitting all timestamps) is of the form
$a_0 \textit{Nil}\, a_1 \textit{Nil}\, \ldots$ where $a_0a_1\ldots$ is a trace of $G(\mathcal{S})$.
\item[$\textsf{C2}'$] Each $\textit{Nil}$ comes with no delay and no two `write' or `read' actions occur at the same time,
i.e., if $(a_i, \tau)(\textit{Nil}, \tau')(a_{i+1}, \tau'')$ is a substring of the play then $\tau = \tau'$ and $\tau < \tau''$.
\item[$\textsf{C3}'$] Each $m?$ is preceded exactly 1 t.u. earlier by a corresponding $m!$.
\item[$\textsf{C4}'$] Each $m!$ is followed exactly 1 t.u. later by a corresponding $m?$ if there are actions
that occur at least 1 t.u. after it.
\end{itemize}
It is clear that $\textsf{C1}'$ and $\textsf{C2}'$ are enforced by the plant $\mathcal{P}$.
Now let
\[
\Psi_0' =  \psi^\leftarrow \vee (\wFinally \checkr \wedge \neg \hat{\psi}^\rightarrow_\textit{fst}) \vee \hat{\psi}^\rightarrow 
\]
and $\Psi_0 = \Psi_0' \vee \wFinally \textit{Halt}$, where
$\hat{\psi}^\rightarrow_\textit{fst}$ and $\hat{\psi}^\rightarrow$ are
obtained from the corresponding formulas in
Section~\ref{sec:undecidability} by replacing $\theta_0^\rightarrow$
and $\theta_1^\rightarrow$ with their `hatted' counterparts.  In
essentially the same way as before, $\textsf{C3}'$ and $\textsf{C4}'$
are ensured by these formulas, and one can prove that $\mathcal{S}$
has a halting computation if and only if there is a bounded-precision
or bounded-clocks controller for $\Psi_0$ and $\mathcal{P}$ but no
such controller for $\Psi'_0$ and $\mathcal{P}$.  This claim holds in
both the finite- and infinite-word cases, thanks to the self-loop
$s_\textit{halt} \xrightarrow{\textit{End}, \mathbf{true}}
s_\textit{halt}$ in $\mathcal{P}$.

\parag{\SMTL.}
To show that $\BPRS^b_u(\SMTL)$ and $\BCRS^b_u(\SMTL)$ where $b \in \{\star, \omega\}$ are undecidable,
we rewrite formulas $\Psi_0$ and $\Psi_0'$ so that their negations are \SMTL formulas.
\begin{proposition}
For each $\rho \in \Lang(\mathcal{P})$ that ends with $\checkr$, we have
\begin{IEEEeqnarray*}{l}
\rho \models \neg \wFinally ( \varphi_\textit{W} \wedge \wFinally[<1] \hat{\theta}_1^\rightarrow
\wedge \wFinally[\geq 1] \hat{\theta}_0^\rightarrow) \iff \\
\rho \models \wFinally[<1] \hat{\theta}_0^\rightarrow \vee \wFinally \Big(
\wFinally[\geq 1] \hat{\theta}_1^\rightarrow \wedge \big( (\neg \varphi_\textit{W})
\Until (\wFinally[<1] \hat{\theta}_0^\rightarrow) \big)
\Big) 
\end{IEEEeqnarray*}
and
\begin{IEEEeqnarray*}{l}
\rho \models \neg \wFinally ( \varphi_\textit{W} \wedge \wFinally[<1] \hat{\theta}_1^\rightarrow
\wedge \wFinally[> 1] \hat{\theta}_0^\rightarrow) \iff \\
\rho \models \wFinally[\leq 1] \hat{\theta}_0^\rightarrow \vee \wFinally \Big(
\wFinally[\geq 1] \hat{\theta}_1^\rightarrow \wedge  \big( (\neg \varphi_\textit{W}) \Until (\wFinally[\leq 1] \hat{\theta}_0^\rightarrow) \big)
\Big) \,.
\end{IEEEeqnarray*}
\end{proposition}
Let $\Phi_1$ and $\Phi_1'$ be the formulas obtained by
replacing corresponding subformulas in $\Phi_0$ and $\Phi_0'$, respectively.
One can verify that their negations are in \SMTL.

\parag{\CFMTL.} First note that the negations of $\Phi_1$ and
$\Phi_1'$ are in \CFMTL, hence $\BPRS_u^\star(\CFMTL)$ and
$\BCRS_u^\star(\CFMTL)$ are undecidable.  Then, to show that both
$\BPRS^b_d(\CFMTL)$ and $\BCRS^b_d(\CFMTL)$ where
$b \in \{\star, \omega\}$ are undecidable, we rewrite $\Phi_0$ and
$\Phi_0'$ into \CFMTL formulas.  This can be accomplished by the
equivalences in the following proposition.
\begin{proposition}
For each $\rho \in \Lang(\mathcal{P})$, we have
\begin{IEEEeqnarray*}{l}
\rho \models \bigvee_{m \in M} \wFinally \big( m! \wedge \wFinally[\leq 1] (m? \wedge \Next \checkl)
 \wedge \wFinally[\geq 1] (m? \wedge \Next \checkl) \big) \iff \\
\rho \models \wFinally[\geq 1] \checkl \wedge \neg \wFinally \big(\wFinally[>1] \hat{\theta}_0^\leftarrow \wedge \Next (\wFinally[<1] \hat{\theta}_0^\leftarrow) \big) \wedge \bigwedge_{\substack{m, m' \in M \\ m \neq m'}} \neg \wFinally \big((m! \vee m?) \wedge
\wFinally[=1] (m'? \wedge \hat{\theta}_0^\leftarrow)  \big)
\end{IEEEeqnarray*}
and
\begin{IEEEeqnarray*}{l}
\rho \models \bigvee_{m \in M} \wFinally \big( m! \wedge \wFinally[\leq 1] (m? \wedge \Next \checkr)
 \wedge \wFinally[\geq 1] (m? \wedge \Next \checkr) \big) \iff \\
\rho \models \wFinally[\geq 1] (\varphi_\textit{R} \wedge \Next \checkr) \wedge \neg \wFinally \big(\wFinally[>1] \hat{\theta}_0^\rightarrow \wedge \Next (\wFinally[<1] \hat{\theta}_0^\rightarrow) \big) \wedge \bigwedge_{\substack{m, m' \in M \\ m \neq m'}} \neg \wFinally \big((m! \vee m?) \wedge
\wFinally[=1] (m'? \wedge \hat{\theta}_0^\rightarrow)  \big) \,.
\end{IEEEeqnarray*}
\end{proposition}
One can verify that the resulting formulas $\Phi_2$ and $\Phi_2'$ (obtained by replacing corresponding subformulas
in $\Phi_0$ and $\Phi_0'$, respectively) are in \CFMTL.

\section{Proof of Theorem~\ref{thm:openclosedmitl}}\label{app:openclosedmitl}

We give reductions for
$\BPRS_d^b(\textsf{Open-}\MITL^{\mathrm{ns}}[\Finally[\infty]])$
($= \BPRS_u^b(\textsf{Closed-}\MITL^{\mathrm{ns}}[\Finally[\infty]])$)
as well as
$\BPRS_d^b(\textsf{Closed-}\MITL^{\mathrm{ns}}[\Finally[\infty]])$
($= \BPRS_u^b(\textsf{Open-}\MITL^{\mathrm{ns}}[\Finally[\infty]])$)
where $b \in \{\star, \omega\}$, and the corresponding \BCRS problems.
Without loss of generality, we consider a DCM
$\mathcal{S} = \langle S, s_0, s_\textit{halt}, M, \Delta \rangle$
where $s_0 \neq s_\textit{halt}$, $s_0$ has an outgoing `write'
action, and $s_\textit{halt}$ is the only state in $S$ with no
outgoing transition. 
In what follows, we will use the plant $\mathcal{P}$ and the subformulas $\hat{\theta}_0^\rightarrow$, $\hat{\theta}_1^\rightarrow$
and $\hat{\theta}_0^\leftarrow$ (defined earlier in Appendix~\ref{app:smtlcfmtl})
and let $\hat{\theta}^\rightarrow_2 = \varphi_\textit{WR} \wedge \Big(\varphi_\textit{WR} \wUntil \big(\textit{Nil} \wedge (\textit{Nil} \wUntil \hat{\theta}^\rightarrow_1) \big)\Big)$.

\parag{$\textmd{\textsf{Closed-}}\MITL^{\mathrm{ns}}[\Finally[\infty]]$ as the desired specification.}
In this case, we can use the encoding based on $\textsf{C1}'$--$\textsf{C4}'$ as given in Appendix~\ref{app:smtlcfmtl}. 
We know that $\textsf{C1}'$ and $\textsf{C2}'$ are enforced by $\mathcal{P}$.
Now we give the formulas for the other conditions:
\begin{enumerate}

\item ($\textsf{C3}'$):
\begin{IEEEeqnarray*}{rCl}
\eta_1 & = & \bigvee_{m \in M} \wFinally \Big( m! \wedge \wFinally[\leq 1] \big(m? \wedge (m? \wUntil \checkl)\big)
 \wedge \wFinally[\geq 1] \big(m? \wedge (m? \wUntil \checkl)\big) \Big) \,.
\end{IEEEeqnarray*}

\item ($\textsf{C4}'$):
\begin{IEEEeqnarray*}{rCl}
\eta_2 & = & \neg \wFinally ( \phi_\textit{W} \wedge \wFinally[<1] \hat{\theta}_{1}^\rightarrow
\wedge \wFinally[> 1] \hat{\theta}_{0}^\rightarrow ) \wedge \neg \bigvee_{\substack{m, m' \in M \\ m \neq m'}} \wFinally \big( m! \wedge \wFinally[<1] \hat{\theta}_{2}^\rightarrow
\wedge \wFinally (m'! \wedge \hat{\theta}_{1}^\rightarrow) \wedge \wFinally[> 1] \hat{\theta}_{0}^\rightarrow \big) \,.
\end{IEEEeqnarray*}

\end{enumerate}
The overall formulas are
\begin{IEEEeqnarray*}{rCl}
\Psi_2' & = & \eta_1 \vee \eta_2 \\
\Psi_2  & = & \eta_1 \vee \eta_2 \vee \wFinally \textit{Halt} \,.
\end{IEEEeqnarray*}
Now we can state a proposition similar to Proposition~\ref{prop:realizability}.
Indeed, the only anomaly that may go undetected is when there is an $m!$ at time $t$
and the controller reaches $s_\textit{halt}$ by an $m'!$ at $t + 1$ with $m \neq m'$;
however in that case, the controller may as well propose $m'!$ at $t + 1 - \epsilon$
(for some $\epsilon > 0$) and reach $s_\textit{halt}$ in a way that respects the encoding.
\begin{proposition} $\mathcal{S}$ has a halting computation if and
  only if there is a bounded-precision or bounded-clocks controller
  for $\Psi_2$ and $\mathcal{P}$ but no such controller for $\Psi'_2$
  and $\mathcal{P}$.
\end{proposition}
Finally, note that we can replace, e.g., $\textit{Nil} \wUntil[>0] \hat{\theta}_0^\rightarrow$ by
\[
\textit{Nil} \wedge \wFinally[>0] \hat{\theta}_0^\rightarrow \wedge \neg \wFinally \big(\varphi_\textit{WR} \wedge \wFinally (\textit{Nil}
\wedge \wFinally \hat{\theta}_0^\rightarrow)\big)
\]
since  $\hat{\theta}_0^\rightarrow$ will happen at most once;
and we can replace, e.g., $\wFinally[\leq 1] \big(m? \wedge (m? \wUntil \checkl)\big)$ by
\[
\wFinally \big(m? \wedge (m? \wUntil \checkl)\big) \wedge \neg \wFinally[>1] \big(m? \wedge (m? \wUntil \checkl)\big)
\]
since $\big(m? \wedge (m? \wUntil \checkl)\big)$ will happen at most once.

\parag{$\textmd{\textsf{Open-}}\MITL^{\mathrm{ns}}[\Finally[\infty]]$ as the desired specification.}
It is known that open timed automata accepts $d$-open sets of timed words, i.e.,
whenever they accept a timed word, they also accept all neighbouring timed words
that are sufficiently `close' to that timed word (see~\cite{OuaWor03} for details).
Recall that $\textsf{C3}'$ asserts that each $m!$ is followed by a corresponding $m?$
at exactly 1 t.u. later; this condition is clearly not $d$-open.
Indeed, it is not possible to give an $\textmd{\textsf{Open-}}\MITL^{\mathrm{ns}}[\Finally[\infty]]$ formula
$\phi^\textit{open}$ such that 
\begin{IEEEeqnarray*}{rCl}
\phi^\textit{open} & \models & (m!, 0)(\textit{Nil}, 0)(m?, 1)(\checkl, 1) \\
\phi^\textit{open} & \not \models & (m!, 0)(\textit{Nil}, 0)(m?, 1-\varepsilon)(\checkl, 1-\varepsilon)
\end{IEEEeqnarray*}
for any  $\varepsilon > 0$.
Since $\mathcal{P}$ does not affect the timing of events, we have to switch to a $d$-open encoding
where we use $\textsf{C3}''$ and $\textsf{C4}''$ in place of $\textsf{C3}'$ and $\textsf{C4}'$:
\begin{itemize}
\item[$\textsf{C3}''$] Each $m?$ at $t$ is preceded by a corresponding $m!$ at $t' \in (t''-1, t-1)$ where
($i$) $t'$ is the maximal timestamp in $(t''-1, t-1]$, ($ii$) $t''$ is the timestamp of the first `write' or `read' action before $m?$.
\item[$\textsf{C4}''$] Each $m!$ at $t$ is followed, if there are actions that occur at time $\geq t + 1$,
by a corresponding $m?$ at $t' \in (t + 1, t'' + 1)$
where ($i$) $t'$ is the minimal timestamp in $[t + 1, t'' + 1)$,
($ii$) $t''$ is the timestamp of the first `write' or `read' action after $m!$.
\end{itemize}
It is clear that $\textsf{C1}'$ and $\textsf{C2}'$ are enforced by the plant $\mathcal{P}$.
The rest of the conditions will be ensured by the following formulas:
\begin{enumerate}

\item ($\textsf{C3}''$):
\begin{IEEEeqnarray*}{rCl}
\beta_1 & = & \bigvee_{m \in M} \wFinally \Big(
m! \wedge \wFinally[>1] (m? \wedge \hat{\theta}_0^\leftarrow)
\wedge \neg \wFinally[\geq 1] (\textit{Nil} \wedge \wFinally \hat{\theta}_0^\leftarrow)
\wedge \neg \wFinally \big(\textit{Nil} \wedge \wFinally (\varphi_\textit{WR} \wedge \wFinally[\geq 1] \hat{\theta}_0^\leftarrow) \big)
\Big) \,.
\end{IEEEeqnarray*}

\item ($\textsf{C4}''$):
\begin{IEEEeqnarray*}{rCl}
\beta_\textit{fst}^\rightarrow & = & \wFinally \big( \varphi_\textit{W} \wedge \wFinally[\geq 1] \hat{\theta}_0^\rightarrow \wedge \neg \wFinally[>1] (\textit{Nil} \wedge \wFinally \hat{\theta}_0^\rightarrow) \big) \\
\beta^\rightarrow & = & \neg \wFinally \Big( \varphi_\textit{W} \wedge \wFinally[\geq 1] \hat{\theta}_0^\rightarrow
\wedge \neg \wFinally[> 1] (\textit{Nil} \wedge \wFinally \hat{\theta}_0^\rightarrow)
\wedge \wFinally \big(\textit{Nil} \wedge \wFinally (\varphi_\textit{WR} \wedge \wFinally[\geq 1] \hat{\theta}_0^\rightarrow) \big)
 \Big)
\wedge \beta_1[\checkr/\checkl] \,.
\end{IEEEeqnarray*}
\end{enumerate}
The overall formulas are
\begin{IEEEeqnarray*}{rCl}
\Psi_1' & = & \beta_1 \vee \big(\wFinally \checkr \wedge (\beta_\textit{fst}^\rightarrow \implies \beta^\rightarrow) \big)  \\
\Psi_1 & = & \beta_1 \vee \big(\wFinally \checkr \wedge (\beta_\textit{fst}^\rightarrow \implies \beta^\rightarrow) \big) \vee \wFinally \textit{Halt} \,.
\end{IEEEeqnarray*}
\begin{proposition} $\mathcal{S}$ has a halting computation if and
  only if there is a bounded-precision or bounded-clocks controller
  for $\Psi_1$ and $\mathcal{P}$ but no such controller for $\Psi'_1$
  and $\mathcal{P}$.
\end{proposition}
Finally, note that we can replace all `until' subformulas as before.

\section{A 3\EXP\ algorithm for $\BRRS^\star_d(\MITL)$}\label{sec:3exp-algor-brrsst}

In this appendix, we present with more details the algorithm for
$\BRRS^\star_d(\MITL)$ that we have briefly sketched in
Section~\ref{sec:BoundReacSynt}. An even more exhaustive presentation
of the algorithm can be found in~\cite{Est15}.

Our decision procedure is based on the translation of \MITL formulas into
equivalent \emph{one-clock alternating timed automata} (\OCATA), a
model of timed automata with a single clock but transitions
dynamically forking to check several properties in parallel. A
\emph{one-clock alternating timed automaton} (shortly, \OCATA) over the
alphabet $\Sigma$ is a tuple $\A = (L,\ell_0,\atransitions,L_f)$ where
$L$ is a finite set of locations, $\ell_0\in L$ is the initial
location, $L_f\subseteq L$ is a set of final locations, and
$\atransitions\colon L\times \Sigma \to \Operations(L)$, where
$\Operations(L)$ is the set of formulas $\operation$ defined by
\[\operation := \top \mid \bot \mid \ell \mid y.\operation \mid x
\bowtie c \mid \operation\lor \operation \mid \operation\land
\operation\]
where $\ell\in L$, $y$ is the unique clock,
${\bowtie}\in\{{<},{\leq},{=},{>},{\geq}\}$ and $c\in \N$.
Intuitively, $x.\operation$ means that clock $x$ must be reset to $0$.
Whereas disjunctions denote classical non-determinism, conjonctions
moreover denote that two objectives must be fulfilled in the sequel.
Indeed, operations of $\Operations(L)$ may always be equivalently
written in disjunctive normal form.

For instance, consider the \OCATA depicted in
\figurename~\ref{fig:OCATA}. It consists of two locations $\ell_0$
(initial and final) and $\ell_1$. The transition function is defined
by: $\atransitions(\ell_0,b) = \ell_0$,
$\atransitions(\ell_0,a) = \ell_0 \land x.\ell_1$,
$\atransitions(\ell_1,a) = \ell_1$, and
$\atransitions(\ell_1,b) = y\leq 1$. Intuitively, when reading an $a$
from $\ell_0$, two copies are created: one continues in location
$\ell_0$, and another, with a fresh clock $x$ reset to $0$, goes to
$\ell_1$. Location $\ell_1$ is used to ensure that a letter $b$ is
read at most 1 time unit after this splitting. Therefore, this \OCATA
recognises the same timed language as the MITL formula
$\Globally(a\Rightarrow \Finally[\leq 1]b)$. See
\cite{OuaWor07,BriEst13} for a thorough definition of the semantics.

As already said, all \MITL formulas may be inductively translated into
equivalent \OCATA: this is the base step to the use of well-quasi
order techniques to obtain the decidability (with a non primitive
recursive complexity) of $\BRRS^\star_d(\MTL)$
in~\cite{BouBoz06}. Using the fact that the \OCATA coming from \MITL
formulas can be translated into non-deterministic timed automata, the
decidability of the problem $\BRRS^\star_d(\MITL)$ is more directly
obtained by the result of \cite{DSoMad02}. Combining the blow-up in
the size of the non-deterministic timed automaton equivalent to an
\MITL formula and the 2-$\EXP$ complexity of the \BRRS realisability
check of \cite{DSoMad02} for timed automata, we obtain that
$\BRRS^\star_d(\MITL)$ can be decided in 3-$\EXP$. However, this
technique requires the construction and determinisation of a full
region automaton, which prevents its use in practice.

Therefore, our second contribution is to give an \emph{on-the-fly}
algorithm to solve the problem \BRRS over \MITL, yet keeping a
3-$\EXP$ theoretical upper-bound, that avoids the construction and
determinisation of the full region automaton. 
The main idea is to use the interval semantics of \cite{BriEst13} for
the \OCATA obtained from the formula. The classical semantics of an
\OCATA is defined in terms of states (a pair composed of a location and
a valuation of the unique clock $y$ of the \OCATA), and configurations
(a set of states). Instead, in \cite{BriEst13} is introduced the
alternative \emph{interval semantics}, that allows for pairs of
location and \emph{interval}. Intuitively, several instances of the
same location with different valuations are \emph{merged} together, to
gain in concision.

\begin{example}\label{ex:OCATA}
  For the \OCATA of \figurename~\ref{fig:OCATA}, one possible
  configuration in the interval semantics consists of
  $C=\{(\ell_0,[0.05,0.5]),(\ell_1,[0.1,0.4]), (\ell_1,[0.5,0.9])\}$.
  In this configuration, reading a letter $b$ after a delay of $0.1$
  time units will make disappear the two last pairs, since the guard
  $y\leq 1$ over the transition exiting from $\ell_1$ is entirely
  fulfilled in both intervals after translation of delay $0.1$. For
  the first pair, reading a $b$ will simply translate the current
  interval. Therefore, we have the sequence
  $C\xrightarrow{0.1,b} \{(\ell_0,[0.15,0.6])\}$. Reading letter $a$
  after a delay of $0.1$ would in contrary keep the two last pairs, by
  translating the intervals, and split the first pair: one goes back
  to $\ell_0$ after translation, the second copy goes to $\ell_1$ with
  a fresh copy of clock $y$. Then, there is a choice: either merging
  the new copy with the next interval associated to location $\ell_1$,
  or keeping a singleton interval. Therefore, we have two possible
  transitions:
  $C\xrightarrow{0.1,a} \{(\ell_0,[0.15,0.6]),(\ell_1,[0,0.5]),
  (\ell_1,[0.6,1])\}$
  and
  $C\xrightarrow{0.1,a}
  \{(\ell_0,[0.15,0.6]),(\ell_1,[0,0]),(\ell_1,[0.2,0.5]),
  (\ell_1,[0.6,1])\}$.
\end{example}

By applying the translation of \cite{OuaWor07}, from every \MITL
formula $\phi$, we can build an equivalent \OCATA $\A_\phi$. Moreover,
a good merging function is described in~\cite{BriEst13} that permits
to keep the same timed language: it non-deterministically merges
intervals in case there are too many (more than a given constant
$M(\phi)$) associated with the same location. In particular, this
allows authors of \cite{BriEst13} to produce directly a
non-deterministic timed automaton $\B_\phi$ that uses $O(M(\phi))$
clocks (one for each endpoint of the intervals), and is equivalent to
$\phi$. 

We now explain our solution for $\BRRS^\star_d(\MITL)$. Therefore, we
fix a formula $\phi$ of \MITL, a plant
$\Plant=(Q,q_0,\ptransitions,Q_f)$ with clocks $X_\Plant$, and a
granularity $\mu=(X,m,K)$ such that $X\cap X_\Plant = \emptyset$. We
suppose built the \OCATA $\A_{\lnot\phi}=(L,\ell_0,\atransitions,L_f)$
over the clock $x\notin X_\Plant\cup X$, and the merging function
associated with the \emph{negation of the formula}. 

\parag{Unfolding the system.} Intuitively, our synthesis problem
can be solved by considering an infinite tree which unfolds all the
possible parallel executions of $\Plant$, $\A_{\lnot\phi}$, and all
the possible controllers. This tree should contain information on the
configuration of the system (plant, \OCATA and controller). Precisely,
\emph{configurations} of the system are tuples $\gamma = (q,E)$ where
$q\in Q$, and $E$ is a finite set of tuples $(\nu,\nu_\Plant,C)$, with
$\nu$ a valuation of the clocks $X$ of the controller, $\nu_\Plant$ a
valuation of the clocks $X_\Plant$ of the plant, and $C$ a
configuration of $\A_{\lnot\phi}$ with at most $M(\phi)$ intervals
associated to each location. We then describe the \emph{dynamics} of
the system. In each configuration, symbolic letters that can be played
are in
$\Gamma = \Sigma \times \AtomicGuards_{m,K}(X\uplus X_\Plant)\times
\powerset{X}$:
it is composed of an action in $\Sigma$, an atomic guard over the
clocks that can observe the controller (so not on the clocks copies of
$\A_{\lnot\phi}$), and a set of clocks chosen by the controller to be
reset. For a symbolic letter $(a,g,R)\in \Gamma$, and configurations
$(q,E)$ and $(q',E')$, we write $(q,E)\xrightarrow{a,g,R} (q',E')$ if
there is a transition of the plant
$(q, (a,g_\Plant,R_\Plant), q')\in \dtransitions_\Plant$, and an
(atomic) guard of the controller's clock
$g_c \in \AtomicGuards_{m,K}(X)$ such that:
\begin{inparaenum}
\item $g = g_c\land g_\Plant$;
\item
  $E' = \{(\nu',\nu'_\Plant,C')\mid \exists t\in\R^+,
  (\nu,\nu_\Plant,C)\in E \text{ such that } \nu+t\models g_c,
  \nu_\Plant+t\models g_\Plant, \nu'=(\nu+t)[R\leftarrow 0],
  \nu'_\Plant = (\nu_\Plant + t)[R_\Plant\leftarrow 0], \text{ and } C
  \xrightarrow{t,a} C'\}$.
\end{inparaenum}
Since the plant is time-deterministic, for all $(q,E)$ and $(a,g,R)$,
there is at most one configuration $(q',E')$ such that
$(q,E)\xrightarrow{a,g,R} (q',E')$.

\parag{Making the tree finite.} The previous tree, starting from the
unique initial configuration
$\gamma_0 = \big(q_0, \{([X\mapsto 0], [X_\Plant\mapsto 0], \{(\ell_0,
[0,0])\})\}\big)$,
unfolding all the executions of the system, is infinitely branching,
because of the density of time, and has infinite depth. Here is how we
cope with these difficulties. First, to deal with the density of time,
we use the classical \emph{region equivalence} \cite{OuaWor07}: we
choose regions that are compatible with the precision $(m,K)$, we
label the nodes of the tree by (sets of) regions, and we unfold
symbolic actions based on the precision $(m,K)$. Moreover, another
source of unboundedness relates with the \OCATA $\A_{\lnot\phi}$ that
may create unboundedly many clock copies during its execution: we rely
on the interval semantics and the merging function to cope with that
issue, which is equivalent to fixing a finite set of clock copies
$\{x_1,x'_1, x_2,x'_2,\ldots, x_{M(\phi)},x'_{M(\phi)}\}$ for
$\A_{\lnot\phi}$. Therefore, labels of the nodes are approximated with
\emph{region words}, only keeping relevant information in terms of
regions and order of the fractional parts of the clocks. The alphabet
used for the letters of the region words will be
$\Lambda = \powerset{(X\uplus X_\Plant)\times \Regions_{m,K} \cup
  (L\times \Regions_{m,K} \times \{1,2,\ldots,M(\phi)\})}$.
Pairs $(y,r)$ of $(X\uplus X_\Plant)\times \Regions_{m,K}$ represents,
the fact that the valuation of clock $y$ is in region $r$. Triples
$(\ell, r, n)$ are now used to describe that in the configuration of
$\A_{\lnot\phi}$, there exists an interval (of index $n$ smaller than
the number $M(\lnot\phi)$ of intervals allowed in the interval
semantics) with one of its bound in region $r$, associated with
location $\ell$. A configuration $\gamma = (q,E)$ with
$E = \{(\nu_j,\nu_{\Plant,j},C_j)\mid 1\leq j\leq J\}$ is then
symbolically represented by the location $q$, and a set of region
words $\wordreg(\eta_j)$ over the alphabet $\Lambda$, for each tuple
$\eta_j = (\nu_j,\nu_{\Plant,j},C_j)$, defined as follows. We first
associate to the tuple $\eta_j$ the set
\begin{align*}
  G &= \{(y,\nu_j(x))\mid y\in X\} \cup \{(y,\nu_{\Plant,j}(x))\mid
      y\in X_\Plant\} \\ 
    &\quad \cup \bigcup_{\ell\in L}\big\{(\ell, \inf(I_1), 1),(\ell,
      \sup(I_1), 1), \ldots, (\ell, \inf(I_m), m),(\ell, \sup(I_m), m) \\[-4mm]
    & \hspace{3cm} \,\mid C_j(\ell) = \{I_1,\ldots,I_m\}\big\} \,.
\end{align*}
The set $G$ is then partitioned into a sequence of subsets
$G_1,\ldots, G_p$ depending on the fractional parts: for all $1\leq
i,j\leq p$, for every pair $(y,u)$ or triple $(\ell,u,k)$ of $G_i$,
and pair $(z,v)$ or triple $(\ell',v,k')$ of $G_j$, we have $i\leq j$
if, and only if, $\fract(u) \leq \fract(v)$. Then, $\wordreg (\eta_j)$ is the
finite word from $\Lambda^\star$ given by $Abs(G_1)\cdots Abs(G_p)$,
where for all $1\leq i\leq p$, 
\[Abs(G_i) = \{\tuple{y,\reg(u)}\mid (y,u)\in G_i\} \cup
\{\tuple{\ell,\reg(u),k} \mid (\ell,u,k)\in G_i\}\,.\]

Two configurations $\gamma$ and $\gamma'$ are equivalent, written
$\gamma\sim \gamma'$, if and only if
$\setwordreg(\gamma) = \setwordreg(\gamma')$. This relation is a
bisimulation with respect to the transition relation defined above:
\begin{lemma}\label{lem:ocata-bisimulation}
  For all configurations $\gamma_1\sim\gamma_2$, $(a,g,R)\in\Gamma$,
  and $\gamma'_1$ such that $\gamma_1\xrightarrow{a,g,R}\gamma'_1$,
  there exists $\gamma'_2$ such that $\gamma'_1\sim\gamma'_2$ and
  $\gamma_2\xrightarrow{a,g,R}\gamma'_2$. 
\end{lemma}

This allows us to lift the transition relation to sets of word
regions: we let $\setwordreg \xrightarrow{a,g,R} \setwordreg'$ if
there exists configurations $\gamma$ and $\gamma'$ such that
$\setwordreg(\gamma)=\setwordreg$,
$\setwordreg(\gamma')=\setwordreg'$, and
$\gamma\xrightarrow{a,g,R} \gamma'$. Therefore, we unfold the
transition relation over sets of word regions, starting from a root
labelled by $\setwordreg(\gamma_0)$ (recall that
$\gamma_0 = \big(q_0, \{([X\mapsto 0], [X_\Plant\mapsto 0], \{(\ell_0,
[0,0])\})\}\big)$ is the unique initial configuration).

Region words guarantee that the tree is finitely branching, yet it
could still have infinite branches. Contrary to the algorithm of
\cite{BouBoz06} that deals with \MTL formulas (and thus general
\OCATA), we do not need to use well-quasi order techniques to cut
branches of the tree. Indeed, labels of the nodes of the tree are now
taken from a finite alphabet, and we simply stop the exploration of a
branch if a node $u$ has the same label as one of its ancestor $u'$:
intuitively, if $u'$ is declared winning, the controller will win if
he plays in $u'$ as he reacts in $u$. We also stop the exploration
when we reach a deadlock, or a node whose label---that we will call
\emph{losing label}---contains at the same time a final location of
$\Plant$, and one of the region word where every location of
$\A_{\lnot\phi}$ is final: this represents a final configuration of
the \OCATA, and models the violation of the specification
$\phi$. Because there is only a finite number of possible labels, this
ensures that the tree $\tau$ hence constructed is finite.

\parag{Turning the finite tree into a game.} We finally define the
game considered on this tree, modeling the realisability question. A
strategy for the controller is a mapping from each node labelled
$\setwordreg$, that is not a leaf of the tree, towards a valid subset
of symbolic actions available in this node, i.e., a subset
$\alpha\subseteq \Gamma$ verifying the following properties:
\begin{enumerate}
\item for all $(a,g,R)\in\alpha$, there exists $\setwordreg'$ such
  that $\setwordreg\xrightarrow{a,g,R}\setwordreg'$;
\item if
  $\alpha\cap \Sigma_C\times \AtomicGuards_{m,K} (X \uplus X_\Plant)
  \times \powerset{X}\neq\emptyset$,
  i.e., an action of the controller is proposed, then there exists
  $(a,g,R)\in\alpha$ with $a\in\Sigma_C$ such that, for all actions
  $(b,g',R')$ fireable from $\setwordreg$ \emph{before} $(a,g,R)$
  (i.e., such that $g$ is a time-successor of $g'$) and with
  $b\in\Sigma_E$, $\alpha$ contains an action $(b,g',R'')$ fireable
  from $\setwordreg$: notice that not all such actions $(b,g',R')$
  need to appear in $\alpha$, since the controller must retain the
  choice of clock he wants to reset;
\item if
  $\alpha\cap \Sigma_C\times \AtomicGuards_{m,K} (X\uplus X_\Plant)
  \times \powerset{X}=\emptyset$,
  i.e., no actions of the controller are proposed, then, for all
  actions $(b,g',R')$ fireable from $\setwordreg$ with $b\in\Sigma_E$,
  $\alpha$ contains an action $(b,g',R'')$ also fireable
  from~$\setwordreg$.
\end{enumerate}
We declare the tree $\tau$ winning in case there exists a strategy
$\strategy$ in $\tau$ such that no reachable leaf from the root when
following choices of the strategy has a losing label. We can decide if
the finite tree $\tau$ is winning, and in such case compute a winning
strategy, using the classical backward induction technique.

\parag{Correctness.} We show that there is a controller for the
instance of $\BRRS^\star_d(\MITL)$ if and only if there is a winning
strategy for the controller in the corresponding (finite) tree. First,
it is easy to check that, \emph{if} there exists a controller for the
\BRRS problem, \emph{then} we can extract a winning strategy from
it. Therefore, we only give a sketch of the other implication:
\emph{if} the tree is winning, \emph{then} there exists a controller
for the \BRRS problem. Consider that the tree $\tau$ is winning, and
call $\pi$ a winning strategy. Intuitively, the locations of the
controller we extract from $\tau$ are labelled by $\tau$'s nodes, and,
for all controller locations labelled by node $u$, the outgoing
transitions should correspond to the winning set of valid actions
$\pi(u)$. Unfortunately, the situation is not so simple, because the
node $u$ could be a leaf that has not been developed because of an
ancestor $u'$ with the same label. In this case, the controller simply
mimics in $u$ the decision taken in $u'$ and goes up in the tree. All
details settled, we obtain:

\begin{proposition}
  The tree $\tau$ is winning if and only if there exists a controller
  in the $\BRRS^\star_d(\MITL)$ problem, i.e., there exists a
  $(X,m,K)$-granular symbolic alphabet $\Gamma$ based on $(\Sigma,X)$,
  and a time-deterministic STS $\T$ over $\Gamma$ such that
  $T\Sigma^\star_{\T,\Plant}\cap \Lang(\Plant) \subseteq
  \Lang(\phi)\cup\{\varepsilon\}$.
\end{proposition}

As already announced, a first advantage of this technique (contrary to
previous methods of \cite{DSoMad02,BouBoz06}) is that we do not
require to construct the non-deterministic timed automaton equivalent
to the \MITL specification. Another advantage is the possibility to
build the tree on-the-fly, i.e., to return a positive or negative
answer to the realisability check, as soon as we are able to compute
it, without constructing the whole execution tree. Notice moreover
that, following techniques presented, e.g., in \cite{CasDav05}, it may
be possible to conclude very quickly whether or not the tree $\tau$ is
winning, by back-propagating as early as possible information
regarding the winning status of a node: for instance, if a winning
strategy has been found while exploring a node $u$, it might induce a
winning strategy for the parent $u'$ of $u$, inducing that we can stop
the exploration of other children of $u'$. In the worst-case scenario,
the size of the tree will still be bounded by the size of the game
constructed in \cite{DSoMad02} for \MITL which ensures the 3-$\EXP$
theoretical upper-bound for our algorithm. In practice however, the
exploration of the tree might yield much more quickly to a positive or
negative answer for the realisability question.


\begin{thebibliography}{10}

\bibitem{AluDil94}
R.~Alur and D.~L. Dill.
\newblock A theory of timed automata.
\newblock {\em T.C.S}, 126(2):183--235, 1994.

\bibitem{AluFed96}
R.~Alur, T.~Feder, and T.~A. Henzinger.
\newblock The benefits of relaxing punctuality.
\newblock {\em J. ACM}, 43(1):116--146, 1996.

\bibitem{BohyBFJR12}
A.~Bohy, V.~Bruy{\`{e}}re, E.~Filiot, N.~Jin, and J.~Raskin.
\newblock Acacia+, a tool for {LTL} synthesis.
\newblock In \textit{CAV'12}, LNCS 7358, Springer.

\bibitem{BouBoz06}
P.~Bouyer, L.~Bozzelli, and F.~Chevalier.
\newblock Controller synthesis for {MTL} specifications.
\newblock In {\em {CONCUR}'06}, LNCS 4137, Springer.

\bibitem{BouMar07}
P.~Bouyer, N.~Markey, J.~Ouaknine, and J.~Worrell.
\newblock The cost of punctuality.
\newblock In {\em LICS'07}, pages 109--120. IEEE.

\bibitem{Brand83}
D.~Brand and P.~Zafiropulo.
\newblock On communicating finite state machines.
\newblock {\em J. ACM}, 30:323--342, 1983.

\bibitem{BriEst13}
T.~Brihaye, M.~Esti{\'e}venart, and G.~Geeraerts.
\newblock On {MITL} and alternating timed automata.
\newblock In {\em FORMATS'13}, LNCS 8053, Springer.

\bibitem{BriEst16}
T.~Brihaye, M.~Esti{\'e}venart, G.~Geeraerts, H.-M.~Ho, B.~Monmege, and N.~Sznajder.
\newblock Real-time Synthesis is Hard! (full version)
\newblock \url{http://www.ulb.ac.be/di/verif/ggeeraer/papers/synthMITL.pdf}

\bibitem{BulDav14}
P.~E. Bulychev, A.~David, K.~G. Larsen, and G.~Li.
\newblock Efficient controller synthesis for a fragment of
$\MTL_{\mbox{0,{\(\infty\)}}}$.
\newblock {\em Acta Inf.}, 51(3-4):165--192, 2014.

\bibitem{CasDav05}
F.~Cassez, A.~David, E.~Fleury, K.~G. Larsen, and D.~Lime.
\newblock Efficient on-the-fly algorithms for the analysis of timed games.
\newblock In {\em CONCUR'05}, LNCS 3653, Springer.

\bibitem{AlfFae03}
L.~de~Alfaro, M.~Faella, T.~A. Henzinger, R.~Majumdar, and M.~Stoelinga.
\newblock The element of surprise in timed games.
\newblock In {\em CONCUR'03}, LNCS 2761, Springer.

\bibitem{DoyGee09}
L.~Doyen, G.~Geeraerts, J.-F. Raskin, and J.~Reichert.
\newblock Realizability of real-time logics.
\newblock In {\em FORMATS'09}, LNCS 5813, Springer.

\bibitem{DSoMad02}
D.~D'Souza and P.~Madhusudan.
\newblock Timed control synthesis for external specifications.
\newblock In {\em STACS'02}, LNCS 2285, Springer.

\bibitem{Est15}
M. Esti\'evenart.
\newblock Verification and synthesis of MITL through alternating timed
automata.
\newblock PhD. thesis, Universit\'e de Mons, 2015.

\bibitem{FilJin09} 
E.~Filiot, N.~Jin, and J.~Raskin.  \newblock An
  antichain algorithm for {LTL} realizability.  \newblock In {\em
    {CAV}'09}, LNCS 5643, Springer.

\bibitem{Koy90}
R.~Koymans.
\newblock Specifying real-time properties with metric temporal logic.
\newblock {\em Real-Time Systems}, 2(4):255--299, 1990.

\bibitem{OuaWor03}
J.~Ouaknine and J.~Worrell.
\newblock Universality and language inclusion for open and closed timed
  automata.
\newblock In {\em HSCC'03}, LNCS 2623, Springer.

\bibitem{OuaWor06}
J.~Ouaknine and J.~Worrell.
\newblock Safety metric temporal logic is fully decidable.
\newblock In {\em TACAS'06}, LNCS
  3920, Springer.

\bibitem{OuaWor07}
J.~Ouaknine and J.~Worrell.
\newblock On the decidability and complexity of metric temporal logic over
  finite words.
\newblock {\em LMCS}, 3(1), 2007.

\bibitem{PnuRos89}
A.~Pnueli and R.~Rosner.
\newblock On the synthesis of an asynchronous reactive module.
\newblock In {\em ICALP'89}, LNCS 372, Springer.

\bibitem{Raskin1999}
J.-F. Raskin.
\newblock {\em Logics, automata and classical theories for deciding real time}.
\newblock PhD thesis, FUNDP (Belgium), 1999.

\end{thebibliography}
\end{document}